\documentclass[11pt]{article}
\usepackage{fullpage}
\usepackage{url}
\usepackage{xspace}
\usepackage{caption}
\usepackage{subcaption}
\usepackage{graphics}
\usepackage[dvips]{epsfig}

\usepackage{amsmath}
\usepackage{amssymb}
\usepackage{amsfonts}
\usepackage{graphicx}
\usepackage{comment}
\usepackage{multirow}
\usepackage{graphics}
\usepackage{latexsym}
\usepackage{algorithm}
\usepackage{algorithmic}
\usepackage{enumerate}
\usepackage{tikz}
\usepackage{color}
\usepackage{caption}
\usepackage{lipsum}

\newtheorem{theorem}{Theorem}[section]
\newtheorem{lemma}[theorem]{Lemma}
\newtheorem{corollary}[theorem]{Corollary}

\newtheorem{proposition}[theorem]{Proposition}

\newcommand{\qed}{\hfill $\Box$ \bigbreak}
\newenvironment{proof}{\noindent {\bf Proof.}}{\qed}

\newcommand{\cC}{{\cal C}}

\newcommand{\cT}{{\cal T}}

\newcommand{\cB}{{\cal B}}
\newcommand{\cF}{{\cal F}}

\newcommand{\cA}{{\cal A}}

\newcommand{\remove}[1]{}

\begin{document}

\title{{\bf Leader Election in Trees with Customized Advice}}
\date{}
\newcommand{\inst}[1]{$^{#1}$}

\author{
Barun Gorain \footnotemark[1]
\and
Andrzej Pelc \footnotemark[1] \footnotemark[2]
}

\date{ }
\maketitle
\def\thefootnote{\fnsymbol{footnote}}

\footnotetext[1]{
\noindent
 D\'{e}partement d'informatique, Universit\'{e} du Qu\'{e}bec en Outaouais,
Gatineau, Qu\'{e}bec J8X 3X7,
 Canada. E-mails:
{\tt baruniitg123@gmail.com}, {\tt pelc@uqo.ca}}

\footnotetext[2]{
\noindent
Research supported in part by NSERC  Discovery Grant 8136 -- 2013  and by the
Research Chair in Distributed Computing of the
Universit\'{e} du Qu\'{e}bec en Outaouais.
}

\begin{abstract}
Leader election is a basic symmetry breaking problem in distributed computing. All nodes of a network have to agree on a single node, called the leader.
If the nodes of the network have distinct labels, then agreeing on a single node means that all nodes have to output the label of the elected leader.
 If the nodes are anonymous,
the task of leader election is formulated as follows: every node of the network must output a simple path starting at it, which is coded as a sequence of port numbers, such that
all these paths
end at a common node, the leader. In this paper, we study deterministic leader election in anonymous trees.

Our goal is to establish tradeoffs between the allocated time $\tau$ and the amount of information that has to be given {\em a priori} to the nodes of a network to enable leader election in time $\tau$.
Following the framework of {\em algorithms
with advice}, this information is provided to all nodes at the start by an oracle knowing the entire tree, in form of binary strings assigned to all nodes. There are two possible variants of  formulating this advice assignment. Either the strings provided to all nodes are identical, or strings assigned to different nodes may be potentially different, i.e., advice can be {\em customized}. As opposed to previous papers
on leader election with advice, in this paper we consider the latter option.

The maximum length of all assigned binary strings is called the {\em size of advice}.
 For a given time $\tau$ allocated to leader election, we give upper and lower bounds on the minimum size
of advice sufficient to perform leader election in time $\tau$. All  our bounds except one pair are tight up to multiplicative constants,
and in this one exceptional case, the gap between the upper and the lower bound is very small.

\vspace{2ex}

\noindent {\bf Keywords:} leader election, tree, advice, deterministic distributed algorithm, time.
\end{abstract}

\vfill

\vfill

\thispagestyle{empty}
\setcounter{page}{0}
\pagebreak

%%%%%%%%%%%%%%%%%%%%%%%%%%%%%%%%%%%%%%%%%%%%%%%%%%%%%%%%%%%
\section{Introduction}
%%%%%%%%%%%%%%%%%%%%%%%%%%%%%%%%%%%%%%%%%%%%%%%%%%%%%%%%%%%

{\bf Background.}
Leader election is a basic symmetry breaking problem in distributed computing  \cite{Ly}. All nodes of a network have to agree on a single node, called the leader.
Performing leader election is essential in applications where a single node has to subsequently broadcast a message to coordinate some network task,
or when all nodes have to transmit data to a single node.

If the nodes of the network have distinct labels, then agreeing on a single node means that all nodes have to output the label of the elected leader. However, in many
applications, even if nodes have distinct identities, they may refuse to reveal them, e.g., for privacy or security reasons. Hence, it is often important to design leader election algorithms that do not depend on the knowledge of such labels. Thus we model networks as anonymous graphs. Under this scenario, we formulate the leader election problem as in \cite{GMP}: every node has to output a simple path (coded as a sequence of port numbers) from it to a common node.

\noindent
{\bf Model and Problem Description.}
We focus on deterministic leader election algorithms for trees.
The network is modeled as a simple undirected tree with $n$ nodes and diameter $D$.
Nodes do not have any identifiers.
On the other hand, we assume that, at each node $v$,
each edge incident to $v$ has a distinct {\em port number} from
$\{0,\dots,d-1\}$, where $d$ is the degree of $v$. Hence each edge has two corresponding port numbers, one at each of its endpoints.
Port numbering is {\em local} to each node, i.e., there is no relation between
port numbers at  the two endpoints of an edge. Initially, each node knows only its own degree.
The task of leader election is formally defined as follows. Every node $v$ of the tree must output a sequence $P(v)=(p_1,\dots,p_k)$ of nonnegative integers.
For each node $v$, let $P^*(v)$ be the path starting at $v$ that results from taking the number $p_i$ from $P(v)$ as the outgoing port at the $i^{th}$ node of the path.
All paths $P^*(v)$ must be simple paths in the tree that end at a common node, called the leader.

Note that, in the absence of port numbers, there would be no way to identify the elected leader by non-leaders, as all
ports, and hence all neighbors, would be indistinguishable to a node.
The above mentioned security and privacy reasons for not revealing node identifiers are irrelevant in the case of port numbers.

Our aim is to establish tradeoffs between the allocated time and the amount of information that has to be given {\em a priori} to the nodes to enable them to perform
leader election.
Following the framework of {\em algorithms
with advice}, see, e.g.,   \cite{DiPe,EFKR,FGIP,FKL,FPR,GPPR02,GMP}, this information is provided to all nodes at the start by an oracle knowing the entire tree, in the form of binary strings assigned to all nodes. There are two possible variants of  formulating this advice assignment. Either the strings provided to all nodes are all identical \cite{DiPe,FGIP,GMP}, or strings assigned to different nodes may be potentially different, i.e., advice can be {\em customized} \cite{FKL,FPR,GPPR02}. As opposed to previous papers
on leader election with advice \cite{DiPe,GMP}, in this paper we consider the latter option.

For a given tree $T=(V,E)$, the advice assigned by the oracle to nodes of $T$ is formally defined as a function $\cA _T:V  \longrightarrow S$, where $S$ is the set of finite binary strings.
The string $\cA _T (v)$, given by the oracle to node $v$,  is an input to a leader election algorithm. Apart from $\cA _T (v)$, node $v$ knows a priori only its own degree.
The maximum of all lengths of strings  $\cA _T (v)$, for $v \in V$, is called the {\em size of advice}. The size of the range $\{\cA _T (v) : v \in V\}$ is called the {\em valency of advice}.
When valency is 1, the advice given to all nodes is identical. As mentioned above, in this paper we assume that valency of advice is larger than 1.
We consider two scenarios. In the first one, valency is unbounded, i.e., every node can potentially get a different advice string. In the second scenario, we bound valency of possible advice by a constant $\lambda >1$.  This may be important in some applications, as a small number of binary strings may be sometimes easier to distribute among nodes.

We use the well-known $\cal{LOCAL}$ communication model \cite{Pe}. Communication proceeds in synchronous
rounds and all nodes start simultaneously.\footnote{It is  known that the $\cal{LOCAL}$  model can be simulated in an asynchronous network using time stamps. } In each round, each node
can exchange arbitrary messages with all of its neighbors, and perform arbitrary local computations. For any tree $T$, any advice $\cA _T$ given to the nodes of $T$, any node $v$, and
any nonnegative integer $\tau$, we define
the {\em labeled ball} $\cB(v,\tau)$ acquired in $T$ by $v$ within $\tau$ communication rounds. This is all the information that $v$ gets about the tree $T$ in $\tau$ rounds. Thus the
labeled ball $\cB(v,\tau)$ consists of  the port-labeled subtree induced by all nodes at distance at most $\tau$ from $v$, with every node $w$ of the subtree labeled by  $\cA _T (w)$, together with
the degrees of all nodes at distance exactly $\tau$ from $v$. The decisions of $v$ in round $\tau$ in any deterministic algorithm are a function of $\cB(v,\tau)$.
The {\em time} of leader election is the minimum number of rounds sufficient to complete it by all nodes.

For advice of valency larger than 1, leader election is always feasible for some advice of size 1, given sufficient time. Indeed, it is enough to give advice 1 to some node, advice 0 to all others, and allocate time $D$.
By this time, all nodes can see the unique node with advice 1, and hence can find a simple path to it. This is in sharp contrast with the scenario of advice with valency 1
 \cite{DiPe,GMP}, as for such advice,
if the tree is perfectly symmetric (e.g., the two-node tree) then symmetry cannot be broken, and leader election is impossible to carry out, regardless
of the allocated time. However, the difficulty of {\em fast} leader election (in time smaller than $D/2$) with {\em small} customized advice, lies in the ability of the oracle to code the part of the path from the node $v$ to the leader that the node $v$ cannot see within the allocated time $\tau$, by appropriately assigning advice to nodes in the labeled ball  $\cB(v,\tau)$ that the node can see.  This is challenging, as advice given to a node $w$ must be used by all nodes to whose labeled balls node $w$ belongs.
Efficient coding schemes using constant size advice are at the heart of our most involved  election algorithm, the one working for constant valency and large diameter.

One may ask if leader election is possible in a very short time, provided that  sufficiently large advice of some valency larger than 1 is given to the nodes. This is, of course, the case when valency is unbounded: then leader election is possible in time 0, as every node can be simply given as advice a code of a simple path to a chosen leader. However, for bounded valency, this is not the case. Fig. \ref{fig:intro2}
gives an example of a tree in which leader election is impossible in time at most 1, for any 2-valent advice, regardless of its size. Indeed, consider the labeled balls $\cB(v,1)$ of the 5 leaves $A,B,C,D,E$ of this tree.
\begin{figure}[h]
\centering
\includegraphics[width=0.5\textwidth]{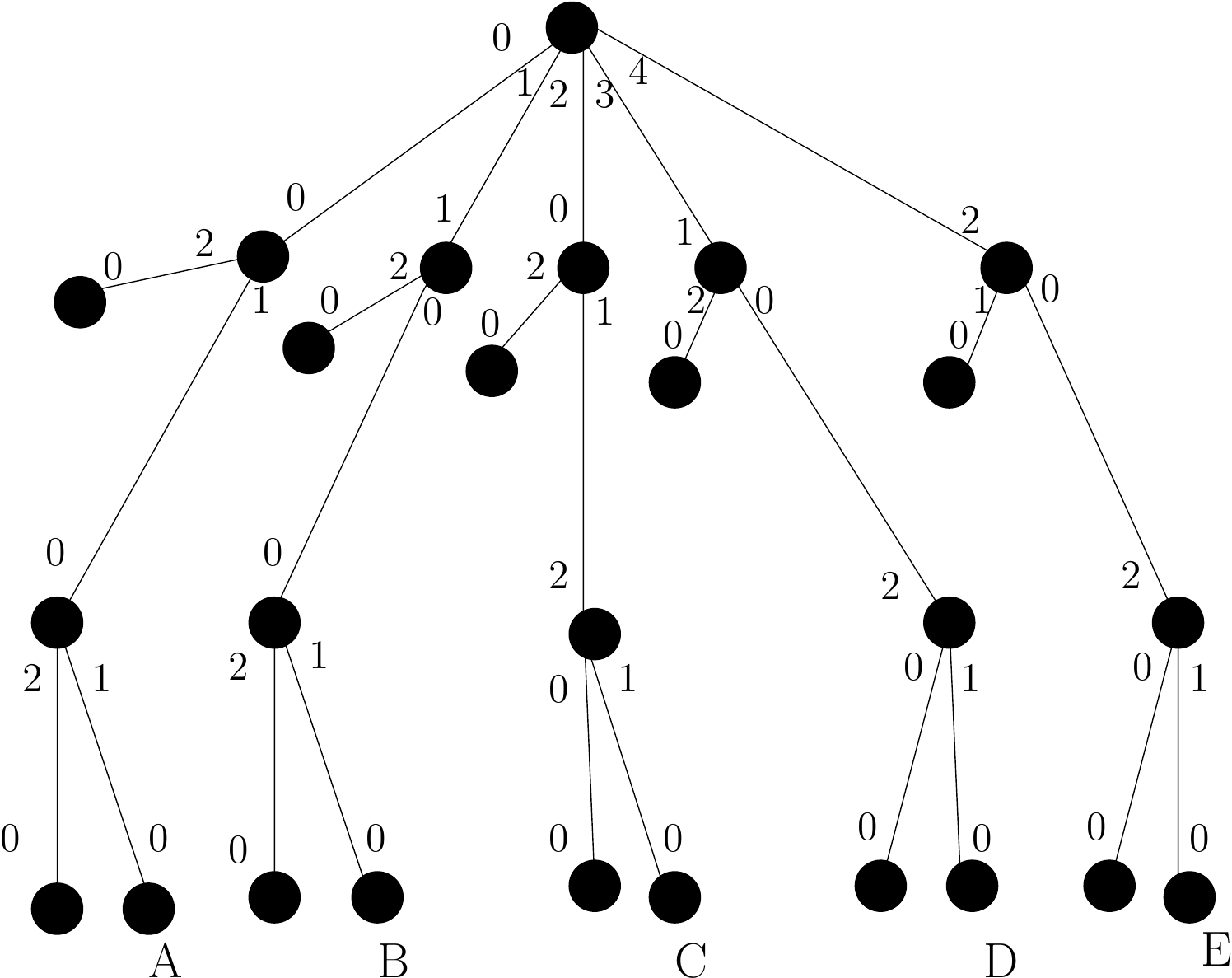}
\caption{An Example of a tree $T$ with $\xi_{2}(T)>1$ }
\label{fig:intro2}
\end{figure}
For any 2-valent advice, there are 4 different possible labeled balls. Hence, at least two of these 5 leaves must have identical labeled balls, and hence must output identical paths to the leader. However, regardless of the choice of the leader in this tree, for one of these leaves this path must be incorrect.

Hence, for any tree $T$ and any integer constant $\lambda>1$, it is important to introduce the parameter $\xi _{\lambda}(T)$, called the $\lambda$-{\em election index} of $T$, defined as the minimum time in which leader election is feasible in $T$,
where the minimum is taken over all possible $\lambda$-valent advice assignments in $T$. Since, as observed above, given sufficient time, leader election is always possible in any tree, if $\lambda >1$,
the $\lambda$-election index is always well defined for $\lambda>1$.

By definition, asking about the minimum size of $\lambda$-valent advice
to solve leader election in time $\tau$ is meaningful only in the class of trees $T$ for which $\xi _{\lambda}(T) \leq \tau$, because otherwise, no $\lambda$-valent advice can help. In the light of these remarks, we are now able to precisely formulate
the two central problems considered in this paper.
\begin{itemize}
\item
For a given time $\tau$, what is the minimum size of advice (of unbounded valency) that permits leader election in time $\tau$ for all trees?
\item
For a given time $\tau$, what is the minimum size of advice of valency $\lambda$ that permits leader election in time $\tau$ for all trees $T$ for which $\xi _{\lambda}(T) \leq \tau$?
\end{itemize}
If the allocated time is at least  $\lceil \frac{D}{2} \rceil$, then advice of size 1 suffices. Indeed, in every tree there is a node whose distance from any other node is at most $\lceil \frac{D}{2} \rceil$, and thus it is enough to give advice 1 to this node and 0 to all others. Hence we concentrate on time smaller than $\lceil \frac{D}{2} \rceil$.

\noindent
{\bf Our results.}
For advice of unbounded valency, we give tight upper and lower bounds on the minimum size of advice sufficient to perform leader election, for the entire spectrum of the allocated time $\tau$.
For the class of $n$-node trees of diameter $D$, the minimum size of advice is  $\Theta(\frac{(D/2) -\tau}{\tau}\log \frac {(n/2)-\tau}{(D/2) -\tau})$ for $0<\tau< \lceil \frac{D}{2} \rceil$, and it is  $\Theta (D\log(n/D))$, for $\tau=0$.

For $\lambda$-valent advice, where $\lambda$ is any integer constant larger than 1,  we give upper and lower bounds on the minimum size of advice sufficient to perform leader election, for a large part of the spectrum of the allocated time $\tau$. Our lower bounds are again tight, except in one case, where the ratio between the upper and the lower bound is smaller than any polynomial. It turns out that, for $\lambda$-valent advice, the minimum size of advice obeys a dichotomy rule: it is either very large or very small. %This is in sharp contrast with the smooth behavior of the minimum size of advice, described above for unbounded valency.
More precisely, we prove the following results.

Consider the class of $n$-node trees of diameter $D$.
\begin{itemize}
\item
If $D$ is small ($D\in\omega(1)$ and $D\in o(\log n)$), and the allocated time is $\tau=\lfloor \alpha D\rfloor$, for some positive constant  $\alpha <1/2$,
then the minimum size of advice is in $\Omega(n^{\delta})$, for any constant $\delta<1$. Since it is in $O(n)$ by \cite{GMP},
the ratio between the upper and the lower bound is smaller than $O(n^{\epsilon})$, for any positive constant $\epsilon$.
This is the only pair of bounds that are not tight.
\item
If $D$ is medium ($D\in \omega(\log n)$ and $D\in o(n)$), and the allocated time is $\tau=\lfloor \alpha D\rfloor$, for some constant  $\alpha <1/2$, then the minimum size of advice is $\Theta(n)$.
\item
If $D$ is large ($D=cn+o(n)$, for some positive constant $c$) then there exist two positive reals $\beta_1 <\beta_2<1/2$, depending only on constants $c$ and $\lambda$,
with the following properties:
\begin{enumerate}
\item
if $\tau=\lfloor \beta D\rfloor$ for any constant $\beta<\beta_1$, then the minimum size of advice is $\Theta(n)$,
\item
if  $\tau=\lfloor \beta' D\rfloor$ for any constant $\beta'>\beta_2$, then the minimum size of advice is constant,
\item
$\beta_2-\beta_1<1/8$, for all constants $c$ and $\lambda$ (the part of the time spectrum not covered by our results is small, regardless of $c$ and $\lambda$).

\end{enumerate}
\end{itemize}

The main challenge in our positive results is the design of  advice and of the accompanying election algorithm, such that the advice given to nodes in the labeled ball $\cB(v,\tau)$ codes the part of the path from node $v$ to the leader that node $v$ cannot see. Every node $v$ must be able to decode this part of the path unambiguously,
although labeled balls may heavily intersect. The main difficulty in our negative results is the construction of trees, for which the labeled ball $\cB(v,\tau)$ is so small and the possible paths unseen by $v$ are so numerous, that the information that can be coded by the advice given to nodes of $\cB(v,\tau)$ is insufficient, for some
nodes $v$, to compute their path to any potential leader.

\noindent
%--------------------------------------------------
{\bf Related work.}
%--------------------------------------------------
The leader election problem was introduced in \cite{LL}. Its study started for rings, in the scenario
where all nodes have distinct labels.
A synchronous algorithm based on label comparisons and using
$O(n \log n)$ messages was given in \cite{HS}.  In \cite{FL} it was proved that
this complexity is optimal for comparison-based algorithms.
%On the other hand, the authors showed
%an algorithm using a linear number of messages but requiring very large running time.
An asynchronous algorithm using $O(n \log n)$ messages was given, e.g., in \cite{P}, and
the optimality of this message complexity was shown in \cite{B}. Deterministic leader election in radio networks was studied, e.g.,
in \cite{JKZ,KP,NO}, and randomized leader election, e.g., in \cite{Wil}. In \cite{HKMMJ}, the leader election problem was
investigated in a model based on mobile agents, for networks with labeled nodes.

In many papers \cite{An,AtSn,ASW,BSVCGS,BV,YK2,YK3} leader election was studied
in anonymous networks. In particular, \cite{BSVCGS,YK3} characterized anonymous message-passing networks in which
leader election can be achieved. In \cite{YK2}, the authors studied
leader election in general networks under the assumption that node labels are
not unique.
%They characterize networks in which this can be done and give an algorithm
%which performs election when it is feasible.
%They assume that the number of nodes of the network is known to all nodes.
In  \cite{FKKLS},  the authors
studied feasibility and message complexity of leader election in rings with possibly
nonunique labels, while, in \cite{DoPe}, the authors constructed algorithms for a
generalized leader election problem in rings with arbitrary labels,
unknown (and arbitrary) size of the ring, and for both
synchronous and asynchronous communication.
%Characterizations of feasible instances for leader election were provided in~\cite{C,CM}.
Memory needed for leader election in unlabeled networks was studied in \cite{FP}.
In \cite{FP1}, the authors investigated the time of leader election in anonymous networks
by characterizing this time in terms of the network size, the diameter of the network, and an additional
parameter called the level of symmetry, which measures how deeply nodes have to inspect the network in order to notice differences in their views of it.
%In \cite{DP1}, the authors studied the feasibility of leader election among anonymous agents that
%navigate in a network in an asynchronous way.

The paradigm of {\em algorithms with advice} that studies how
arbitrary kinds of information (coded as binary strings provided to nodes of the network or to agents) can be used to perform network tasks more efficiently was previously
proposed in \cite{AKM01,DP,EFKR,FGIP,FIP1,FIP2,FKL,FP,FPR,GPPR02,IKP,KKKP02,KKP05,MP,SN,TZ05}.
%If advice is given to nodes then, instead of advice, the term {\em informative labeling schemes} is sometimes used if (like in our scenario) different nodes can get different information.
There are two possible ways of assigning advice to nodes. Either all the binary strings provided to nodes are identical \cite{DiPe,FGIP,GMP}, or strings assigned to different nodes may be potentially different  \cite{FKL,FPR,GPPR02}. In the latter case, instead of advice, the term {\em informative labeling schemes}  is sometimes used.

In this paradigm, the focus is on establishing
the minimum size of advice required to solve
network problems in an efficient way.
 %In \cite{KKP05}, given a distributed representation of a solution for a problem,
%the authors investigated the number of bits of communication needed to verify the legality of the represented solution.
In \cite{FIP1}, the authors compared the minimum size of advice required to
solve two information dissemination problems using a linear number of messages.
In \cite{FKL}, it was shown that advice of constant size given to the nodes enables the distributed construction of a minimum
spanning tree in logarithmic time.
In \cite{EFKR}, the advice paradigm was used for online problems.
In \cite{FGIP}, the authors established lower bounds on the size of advice
needed to beat time $\Theta(\log^*n)$
for 3-coloring cycles and to achieve time $\Theta(\log^*n)$ for 3-coloring unoriented trees.
In the case of \cite{SN}, the issue was not efficiency but feasibility: it
was shown that $\Theta(n\log n)$ is the minimum size of advice
required to perform monotone connected graph clearing.
In \cite{IKP}, the authors studied radio networks for
which it is possible to perform centralized broadcasting in constant time. They proved that constant time is achievable with
$O(n)$ bits of advice in such networks, while
$o(n)$ bits are not enough. In \cite{FPR}, the authors studied the problem of topology recognition with advice given to the nodes.
In \cite{DP}, the task of drawing an isomorphic map by an agent in a graph was considered, and the problem was to determine the minimum advice that has to be given to the agent
for the task to be feasible. In \cite{MP}, the authors investigated the minimum size of advice sufficient to find the largest-labeled node in a graph.

The problem of leader election with advice was previously studied for anonymous networks in \cite{DiPe,GMP}.
The main difference between  these papers and the present paper is that in  \cite{DiPe,GMP} the binary strings provided to nodes were all identical, while in the present paper they
may be potentially different. This is a significant difference: while in the former case advice cannot break symmetry and can be only used to extract existing asymmetries from the network more efficiently,
in our case advice has double role: it can break symmetry and  provide additional information that enables nodes to use it fast to perform leader election.
For example, with the possibility of customizing advice, advice of size 1 is always sufficient to perform leader election in time larger than half of the diameter of the tree, while it was proved in \cite{GMP} that large advice was sometimes needed for such allocated time, if all advice strings had to be identical.

\section{Preliminaries}
For any rooted tree $T$ with root $r$ and diameter $D$, the {\em depth} of a node is defined as its distance from the root $r$. The {\em height} of the tree is the  length of the longest  path from the root to a leaf. For two given nodes $v$ and $w$, if $v$ lies on the unique simple path between $w$ and the root $r$, then $w$ is a {\em descendant }of $v$, and $v$ is an {\em ancestor} of $w$.
Every tree has either a central node or a central edge, depending on whether the diameter of the tree is even or odd.
 If the diameter  $D$ is even, then the central node is the unique node in the middle of every simple
path of length $D$, and if the diameter $D$ is odd, then the central edge is the unique edge in
the middle of every simple path of length $D$.

Since there are $2^{\Theta(n)}$ $n$-node anonymous port-labeled trees (cf. \cite{GMP}), $\lambda$-valent advice of size $O(n)$ is sufficient to perform leader election in any tree $T$, in time $\xi_{\lambda}(T)$, for any $\lambda \geq 1$. Hence, whenever we prove a lower bound  $\Omega(n)$ on the size of advice, it is tight (up to multiplicative constants).

Throughout the paper, by path, we mean a simple path. For a sequence $s$, we denote by $s^R$ the reverse sequence.

We use the abbreviation $1^k$ for a string of $k$ ones and $0^k$ for a string of $k$ zeroes.
We will need to efficiently code sequences of natural numbers in a non-ambiguous way, using binary strings. We will use the following coding. Let $\sigma=(p_1,\dots,p_k)$
be a sequence of natural numbers. Let $p_i^*$ be the binary representation of $p_i$, for $i \leq  k$. The binary code $\sigma ^*$ of $\sigma$  is defined in three steps.
Consider the sequence $(p_1^*,\dots,p_k^*)$. The terms of this sequence are 0,1 and ``,''. \\
Step 1:  replace each 0 by the string of bits 10, and replace each 1 by by the string of bits 11. Notice that each comma is now followed by the bit 1.\\
Step 2:  remove the bit 1 after each comma.\\
Step 3: replace each comma by the bit 0.

%and replace each ``,'' by 0
The resulting binary sequence is $\sigma^*$. For example, the sequence $\sigma=(3,5)$ will be transformed as follows:
$(p_1^*,p_2^*)=(11,101)$. After step 1 the sequence becomes $(1111,111011)$. After step 2 the sequence becomes  $(1111,11011)$. After step 3 it becomes
$\sigma^*= (1111011011)$.

Notice that
the length of the sequence $\sigma^*$ is $2\sum_{i=1}^k(\lfloor \log p_i\rfloor +1)$ because $\sum_{i=1}^k(\lfloor \log p_i\rfloor +1)$ is the sum of the lengths of the binary representations $p_1^*,\dots,p_k^*$ . The transformation $\sigma \rightarrow \sigma^*$ is one-to-one because the sequence $\sigma^*$ can be decoded as follows.

Divide the sequence $\sigma^*$ into consecutive pairs of bits, until a pair $(0,b)$. Decode the obtained preceding pairs by replacing 10 by 0 and 11 by 1.
The obtained sequence is $p_1^*$. Add a comma, and replace the pair $(0,b)$ by $(1,b)$. Repeat the above steps to find the consecutive strings $p_i^*$, for
$2 \leq  i \leq k$. This gives the sequence $(p_1^*,\dots,p_k^*)$. Now $\sigma=(p_1,\dots,p_k)$, where $p_i^*$ is the binary representation of $p_i$, for $i \leq  k$.

The above coding can be easily generalized to the case when $\lambda>2$ colors $c_1,c_2,\dots, c_\lambda$ are used in the coding instead of 0 and 1. Now $p_i^*$
is the $\lambda$-ary representation of $p_i$. The three above  coding steps are changed as follows:\\
Step 1: replace each color $c_i$ by $c_1c_i$.\\
Step 2: remove the color $c_1$ after each comma.\\
Step 3: replace each comma by $c_2$.

The obtained code $\sigma^*$ has length at most $2\sum_{i=1}^k(\lfloor \log_{\lambda} p_i\rfloor +1)$.

\section{Election with advice of unbounded valency}

In this section we give tight upper and lower bounds on the minimum size of advice (of unbounded valency) sufficient to perform leader election in
any time $\tau <\lceil \frac{D}{2} \rceil$.

\subsection{Upper bound}

We first present a leader election algorithm working for any tree of diameter $D$ in time $\tau <\lceil \frac{D}{2} \rceil$, with advice of unbounded valency.
Let $T$ be a rooted $n$-node tree of diameter $D$. If $D$ is even, then the root $r$ is the central node, and if $D$ is odd, the root $r$ is one
of the endpoints of the central edge. This is the node that the algorithm will elect. The height of the tree is $\lceil \frac{D}{2} \rceil$.

At a high level, the idea of the algorithm is to partition every branch of the tree into segments of length $\lfloor\frac{\tau}{2}\rfloor$ and assign advice to the nodes of the segment in such a way that the concatenation of the advice strings in a given segment, read bottom-up, can be decoded as the sequence of port numbers corresponding to the path from the upper endpoint of the segment to the root. Every node can see some entire segment, and thus can output the correct path to the leader. Care should be taken to indicate the upward direction in each segment, as nodes cannot recognize this direction a priori.

We first give a detailed description of Algorithm {\textsc{Advice$(T,\tau)$}} constructing the advice, and then present the details of the leader election algorithm using this advice. If $D\leq 2$, it is straightforward to elect a leader using advice of size 1. In the sequel we assume that $D\geq 3$.

The construction of the advice proceeds as follows.  $P(u,v)$ denotes the path from $u$ to $v$,
defined as the sequence of nodes including $u$ and $v$. For $\tau\leq 1$, the advice given to each node is simply the string of port numbers corresponding to the path from it to the leader. Otherwise,
Algorithm {\textsc{Advice$(T,\tau)$}} computes the advice $A(v)=(M_1(v),M_2(v), M_3(v),C(v))$ for each node $v$.  For every branch $(r,v_1,v_2,v_3, \cdots, v_k)$ from the leader to a leaf,  we define $M_1(r)=3$,  $M_1(v_1)=0$, $M_1(v_2)=1$, $M_1(v_3)=2$, $M_1(v_4)=0$, $M_1(v_5)=1$, $M_1(v_6)=2$, and so on. The component  $M_1(v_i)$ of the advice helps a node to identify the upward direction in the following way. For any node $v$, if $M_1(v)=x$, and if $v'$ is the neighbor of $v$ with $M_1(v')=(x-1) \mod 3$, then the node will identify $v'$ as its parent, i.e., its neighbor  on the path towards the leader.

The component $M_2(v)$ of the advice is used to mark every node at depth $k\lfloor\frac{\tau}{2}\rfloor$, for all $k \ge 1$. We set $M_2(v)=1$ if the depth of $v$ is $k\lfloor\frac{\tau}{2}\rfloor$, for $k \ge 1$, otherwise $M_2(v)=0$.

The component $M_3(v)$ of the advice is used to mark all nodes at depth at least $\lceil \frac{D}{2} \rceil-\tau$. We set $M_3(v)=1$ if the depth of $v$ is at least $\lceil \frac{D}{2} \rceil-\tau$, otherwise $M_3(v)=0$.

The component $C(v)$ of the advice is assigned in the following way. Let $P=(u_0,u_1,\cdots, u _{\lfloor\frac{\tau}{2}\rfloor})$, such that $M_2(u_0)=M_2(u _{\lfloor\frac{\tau}{2}\rfloor})=1$,  and $M_1(v_{i+1})=(M_1(v_i)-1)\mod (3)$, for $i =0,1,\cdots,\lfloor\frac{\tau}{2}\rfloor-1$, i.e, $P$ is a path between two nodes which are at
depths multiple of   $\lfloor\frac{\tau}{2}\rfloor$ and the path is going towards the root $r$.
The binary string $s(w,r)^*$ unambiguously coding the sequence of port numbers $s(w,r)$ that represents the path of length at most $\lceil \frac{D}{2} \rceil-\tau$ from some node $w$ of the tree to the leader $r$, is divided into $\lfloor\frac{\tau}{2}\rfloor$ segments of lengths differing by at most 1, and these segments are given as the component $C(u_i)$ to the nodes  $u_1,\cdots, u _{\lfloor\frac{\tau}{2}\rfloor}$. If the depth of $u _{\lfloor\frac{\tau}{2}\rfloor}$ is less than $\lceil \frac{D}{2} \rceil-\tau$, then $u _{\lfloor\frac{\tau}{2}\rfloor}$ is chosen as $w$. Otherwise,  the node at depth $\lceil \frac{D}{2} \rceil-\tau$ on the path from $u _{\lfloor\frac{\tau}{2}\rfloor}$ to $r$ is chosen as $w$.

Since the advice given to every node should be a binary string, the advice outputted by Algorithm {\textsc{Advice$(T,\tau)$}} is the binary string $\rho(v)$ unambiguously coding
$A(v)=(M_1(v),M_2(v), M_3(v),C(v))$. This can be done by coding $M_1(v)$ on three bits, and all the other components without any change.

\begin{algorithm}
\caption{\textsc{Advice$(T,\tau)$}}
\begin{algorithmic}[1]
\IF{$\tau=0$ or  $\tau=1$}
\STATE{Let $s(v,r)$ be the sequence of port numbers corresponding to the simple path from $v$ to $r$.\\ $A(v) \leftarrow (0,0,0,s(v,r)^*)$, for all $v \in T$. }
\ELSE
\STATE{$M_1(r) \leftarrow 3$. For every path $P=(r,v_1,v_2,\cdots,v_k)$, starting from $r$, $ M_1(v_i) \leftarrow (i-1)\mod 3$. }
\STATE{For every node $v \in T$ at depth $k\lfloor\frac{\tau}{2}\rfloor$, for $k \ge 1$,  $ M_2(v) \leftarrow 1$.\\ For all other nodes $v\in T$, $ M_2(v) \leftarrow 0$. }
\STATE{For every node $v \in T$ at depth at least $\lceil \frac{D}{2} \rceil-\tau$, $M_3(v) \leftarrow 1$.\\ For all other nodes $v \in T$, $ M_3(v) \leftarrow 0$.}
\FOR{ every path $P=(u_0,u_1,\cdots, u _{\lfloor\frac{\tau}{2}\rfloor})$, such that $M_2(u_0)=M_2(u _{\lfloor\frac{\tau}{2}\rfloor})=1$,  and $M_1(v_{i+1})=(M_1(v_i)-1)\mod 3$, for $i =1,2,\cdots,\lfloor\frac{\tau}{2}\rfloor-1$}
\IF{$M_3(u _{\lfloor\frac{\tau}{2}\rfloor})=0$}\label{step:algo1}
\STATE{$w \leftarrow u _{\lfloor\frac{\tau}{2}\rfloor}$}
\ELSE
\STATE{ $w \leftarrow$ the last node on the path $P((u _{\lfloor\frac{\tau}{2}\rfloor}),r)$ such that $M_3(w)=1$.}
\ENDIF
\STATE{Let $s(w,r)$ be the sequence of port numbers corresponding to the simple path from $w$ to $r$. Let $s_1,s_2,\cdots, s_{\lfloor\frac{\tau}{2}\rfloor}$ be substrings of $s(w,r)^*$ of lengths differing by at most 1,
 such that $s(w,r)^*$ is the concatenation $s_1s_2\cdots s_{\lfloor\frac{\tau}{2}\rfloor}$.\\  $C(u_i) \leftarrow s_i$, for $i=1,2\cdots,\lfloor\frac{\tau}{2}\rfloor$.  }\label{step:algo2}
\ENDFOR
\STATE{$C(v) \leftarrow 0$, for all other nodes.}
\STATE{$A(v)\leftarrow(M_1(v),M_2(v), M_3(v),C(v))$, for all $v \in T$.}
\ENDIF
\STATE{Output the binary string $\rho(v)$ unambiguously coding $A(v)$,  for all $v \in T$.}
\end{algorithmic}\label{alg:alg1}
\end{algorithm}

Algorithm {\textsc{Election$(v,\rho(v),\tau)$}} using advice $\rho(v)$ given by Algorithm {\textsc{Advice$(T,\tau)$}} works as follows.
% Every node $v$ chooses $r$ as the leader executing Algorithm {\textsc{Election$(v,A(v),\tau)$}} with the advice  $A(v)$.
Every node $v$ decodes from its advice the terms $M_1(v)$, $M_2(v)$, $M_3(v)$, and $C(v)$ of the sequence $A(v)=(M_1(v),M_2(v), M_3(v),C(v))$.
 If $\tau \le 1$, then the node decodes the sequence of port numbers from $C(v)$ (which corresponds to the path from $v$ to $r$) and outputs it.

 Suppose that $\tau >1$.
 If a node $v$ can see a node $u$ with $M_1(u)=3$, then $v$ outputs the sequence of port numbers corresponding to the simple path from $v$ to $u$. This sequence is seen in the ball $\cB(v,\tau)$. Otherwise, each node $v$ can see in time $\tau$ a path $P=(u_0,u_1,\cdots, u _{\lfloor\frac{\tau}{2}\rfloor})$, such that $u_0$ is an ancestor of $v$, $M_2(u_0)=M_2(u _{\lfloor\frac{\tau}{2}\rfloor})=1$,  and $M_1(v_{i+1})=(M_1(v_i)-1)\mod 3$, for $i =0,1,\cdots,\lfloor\frac{\tau}{2}\rfloor-1$, because, for every node $u$ at depth $k\lfloor\frac{\tau}{2}\rfloor$, we have $M_2(u)=1$. From the advice at nodes $u_0,u_1,\cdots, u _{\lfloor\frac{\tau}{2}\rfloor}$, node $v$ decodes  $C(u_1)$, $C(u_2)$, $\cdots$, $C(u _{\lfloor\frac{\tau}{2}\rfloor})$.
 Next,  node $v$ computes the string $s'$ which is the concatenation of $C(u_1)$, $C(u_2)$, $\cdots$, $C(u _{\lfloor\frac{\tau}{2}\rfloor})$. This string $s'$ unambiguously codes the sequence $\pi(w,r)$ of port numbers corresponding to the path from $w$ to $r$, where $w=u _{\lfloor\frac{\tau}{2}\rfloor}$ if the depth of $u _{\lfloor\frac{\tau}{2}\rfloor}$ is less than $ \lceil \frac{D}{2} \rceil -\tau$, and where $w$ is the node at depth $\lceil \frac{D}{2}\rceil-\tau$ on $P(v,r)$, if  the depth of $u _{\lfloor\frac{\tau}{2}\rfloor}$ is at least $ \lceil \frac{D}{2} \rceil -\tau$.
  Let $\pi(v,w)$ be the sequence of port numbers corresponding to the path from $v$ to $w$, which can be seen in the ball $\cB(v,\tau)$.
 The node $v$ outputs the concatenation of sequences  $\pi(v,w)$ and $\pi(w,r)$.

  \begin{algorithm}
\caption{\textsc{Election$(v,\rho(v),\tau)$}}
\begin{algorithmic}[1]
\STATE{Get the labeled ball  $\cB(v,\tau)$ in time $\tau$.}
\STATE{Decode  the terms $M_1(u)$, $M_2(u)$, $M_3(u)$, and $C(u)$ from the advice $\rho(u)$ at all nodes $u$ in $\cB(v,\tau)$.}
\IF{$\tau=0$ or  $\tau=1$}
\STATE{decode the sequence of port numbers from $C(v)$ and output it. }
\ELSE
\IF{ there exists a node $u$ in $\cB(v,\tau)$, such that $M_1(u)= 3$}
\STATE{ let $\pi$ be the sequence of port numbers corresponding to the simple path from $v$ to $u$. }
\ELSE
\STATE{Let $P=(v,v_1,v_2, \cdots,v_{\tau})$ be the path in  $\cB(v,\tau)$ such that $M_1(v_1)=(M_1(v)-1)\mod 3$ and $M_1(v_{i+1})=(M_1(v_i)-1)\mod 3$, for $i =1,2,\cdots,\tau-1$. Let $j$ be the smallest index such that $v_j$ and $v_{j+\lfloor\frac{\tau}{2}\rfloor}$ are the nodes in $P$ satisfying $ M_2(v_j)= M_2(v_{j+\lfloor\frac{\tau}{2}\rfloor})= 1$. Let $s'$ be the concatenation $C(v_{j+1}) C(v_{j+2}) \cdots  C(v_{j+\lfloor\frac{\tau}{2}\rfloor})$. Let $\pi'$ be the sequence of port numbers coded by $s'$. }
\IF{$M_3(v_{j+\lfloor\frac{\tau}{2}\rfloor})=0$}
\STATE{$w \leftarrow v_{j+\lfloor\frac{\tau}{2}\rfloor}$}\label{step:1}
\ELSE
\STATE{ $w \leftarrow$ the last node on the path $P$ such that $M_3(w)=1$.}\label{step:2}
\ENDIF
\STATE{Let $\pi''$ be the sequence of port numbers corresponding to the simple path from $v$ to $w$ (seen in the ball $\cB(v,\tau)$).}
\STATE{$\pi \leftarrow \pi''\pi'$}
\ENDIF
\STATE{Output $\pi$.}
\ENDIF
\end{algorithmic}\label{alg:alg2}
\end{algorithm}

 The following two lemmas establish an upper bound on the size of advice provided by Algorithm \textsc{Advice$(T,\tau)$}  .

\begin{lemma}\label{lem:path}
Let $s(v,r)$ be the  the sequence of port numbers corresponding to the path $P$ from $v$ to $r$. Then the length of  $s(v,r)^*$ is in $O ((D-2\tau) \log (\frac{n-2\tau}{D-2\tau}))$ for
 every node $v\in T$ of depth at most $ \lceil\frac{D}{2}\rceil-\tau$.
\end{lemma}
\begin{proof}
Let $v$ be a node of $T$ of depth at most $ \lceil\frac{D}{2}\rceil-\tau$.
It is enough to prove that the sum of logarithms of degrees of nodes on the path $P$ is in $O ((D-2\tau) \log (\frac{n-2\tau}{D-2\tau}))$.
Since the depth of $v$ in $T$ is at most $\lceil\frac{D}{2}\rceil-\tau$, there exist at least $\tau$ nodes in $T$ with depth larger than $\lceil\frac{D}{2}\rceil-\tau$.
Also, since the diameter of $T$ is $D$, there exists at least one path of length $\lfloor\frac{D}{2}\rfloor$ with no common node with $P$ other than $r$.
Let $P=(v,u_1,u_2,\cdots,u_{k-1},r)$. Let $d(u)$ denote the degree of node $u$.

Then,
$d(v)+ \sum_{i=1}^{k-1} d(u_i) \le n-\tau-\lfloor\frac{D}{2}\rfloor\le n-2\tau$.

The sum of logarithms of these degrees is $ \log d(v) + \sum_{i=1}^{k-1} \log d(u_i) = \log \left(d(v)\prod_{i=1}^k d(u_i)\right)$. The value of $d(v)\prod_{i=1}^k d(u_i)$ is maximized when $d(v)=d(u_i)= \frac{d(v)+ \sum_{i=1}^k d(u_i)}{k}$ for $1\le i \le k-1$. Hence, this sum of logarithms is at most   $k \log (\frac{n-2\tau}{k}) \le (\lceil\frac{D}{2}\rceil-\tau) \log (\frac{n-2\tau}{\lceil\frac{D}{2}\rceil-\tau}) \in O ((D-2\tau) \log (\frac{n-2\tau}{D-2\tau})) $.
\end{proof}

\begin{lemma}\label{lem:size}
Assume that the diameter $D$ of the tree is at least 3.
The size of the advice given to each node $v$ in $T$  by Algorithm \textsc{Advice$(T,\tau)$}  is in $O(D\log \frac{n}{D})$, when $\tau=0$, and it is in  $O (\frac{D-2\tau}{\tau} \log (\frac{n-2\tau}{D-2\tau}))$, when $\tau>0$.
\end{lemma}
\begin{proof}
First suppose that $\tau\leq1$.
In this case, the advice given to every node $v$ codes the sequence of port numbers corresponding to the path from $v$ to $r$.
By Lemma \ref{lem:path}, the size of the advice is  $O(D\log \frac{n}{D})$.

In the rest of the proof we assume that $\tau \geq 2$.
Consider any path $P=(u_0,u_1,\cdots, u _{\lfloor\frac{\tau}{2}\rfloor})$, such that $M_2(u_1)=M_2(u _{\lfloor\frac{\tau}{2}\rfloor})=1$,  and $M_1(v_{i+1})=(M_1(v_i)-1)\mod 3$, for $i =1,2,\cdots,\lceil\frac{\tau}{2}\rfloor-1$. Let $s$ be the concatenation $C(u_1)C(u_2)\cdots C(u _{\lfloor\frac{\tau}{2}\rfloor})$.  According to steps \ref{step:algo1}-\ref{step:algo2} of Algorithm \ref{alg:alg1}, $s$ is the binary string coding the sequence of port numbers corresponding to the path from a node of depth at most $ \lceil\frac{D}{2}\rceil-\tau$ to $r$. Hence, by Lemma \ref{lem:path}, the length $|s|$ of $s$ is in $O ((D-2\tau) \log (\frac{n-2\tau}{D-2\tau}))$. Since $C(u_i) \le \lceil\frac{|s|}{\lfloor\frac{\tau}{2}\rfloor}\rceil$, therefore, $C(u_i)  \in O (\frac{D-2\tau}{\tau} \log (\frac{n-2\tau}{D-2\tau}))$.
For all other nodes $v$ of $T$, we have $C(v)=0$. Therefore, $C(v)  \in O (\frac{D-2\tau}{\tau} \log (\frac{n-2\tau}{D-2\tau}))$, for all nodes $v$ of $T$. The other three components of $A(v)$ are of constant size. Hence the size of the binary string $\rho(v)$ coding the advice $A(v)$ is in $O (\frac{D-2\tau}{\tau} \log (\frac{n-2\tau}{D-2\tau}))$.
\end{proof}

The following lemma proves the correctness of our election algorithm.

\begin{lemma}\label{lem:correctness}
Every node $v$ of a tree $T$ executing,  in time $\tau$, Algorithm {\textsc{Election$(v,\rho(v),\tau)$}} with advice $\rho(v)$ given by Algorithm {\textsc{Advice$(T,\tau)$}},
chooses node $r$ as the leader and outputs the sequence of port numbers corresponding to the path from $v$ to $r$.
 \end{lemma}

 \begin{proof}
If  there exists a node $u$ in $\cB(v,\tau)$, such that $M_1(u)= 3$, i.e., if a node $v$ can see the node $u=r$, then $v$ outputs the sequence of port numbers
corresponding to the  path from $v$ to $r$, reading it from $\cB(v,\tau)$. Otherwise, consider two cases.

 If $\tau \le 1$, the advice assigned to each node codes the sequence of port numbers corresponding to the path from this node to $r$. If $\tau >1$, then the components $C(v)$ of the advice given to nodes $v$ have the following property. Their concatenation, for nodes in a segment of length $\lfloor\frac{\tau}{2}\rfloor$, between nodes at depths which are multiples of $\lfloor\frac{\tau}{2}\rfloor$,  read bottom-up, can be decoded as the sequence of port numbers corresponding to the path from  node $w$
 defined in steps \ref{step:1}-\ref{step:2} to the root $r$. Every node can see at least one such entire segment, recognizes the direction bottom-up, and can see  this node $w$. Therefore, every node $v$ can output the sequence of port numbers corresponding to the path from $v$ to $r$.
 \end{proof}

Lemmas \ref{lem:size} and \ref{lem:correctness} imply  the following theorem.
\begin{theorem}\label{th:ub}
For any $n$-node tree with diameter $D$, Algorithm {\textsc{Election$(v,\rho(v),\tau)$}} performs election in time $\tau$ with advice of size $O(D\log \frac{n}{D})$, when $\tau=0$, and with advice of size $O (\frac{D-2\tau}{\tau} \log (\frac{n-2\tau}{D-2\tau}))$ when $0<\tau<\lceil \frac{D}{2} \rceil$.
\end{theorem}

\subsection{Lower Bound}\label{sec:lb}

Let  $n'>D \ge 3$. Let $L$ be a line of length $D$ with nodes $v_0$, $v_1$, $\cdots$, $v_D$ from left to right, and with port numbers 0 and 1 at each edge from left to right.
Let $z= \lceil\frac{n'-2\tau}{2(\lceil\frac{D}{2}\rceil-\tau)}\rceil$. We construct an $n$-node tree $T$ from $L$ as follows, see Fig. \ref{T1} and Fig. \ref{T2}. For each $i$, such that  $\tau \le i \le \lceil\frac{D}{2}\rceil-1 $, attach $z-1$ nodes of degree one to each of the nodes $v_i$ and $v_{D-i}$ . The port numbers corresponding to the newly added edges at the nodes on $L$ are $2,\dots, z$.

The total number of nodes in $T$ is given by $n= 2\tau+ 2(\lceil\frac{D}{2}\rceil-\tau)z= 2\tau+ 2(\lceil\frac{D}{2}\rceil-\tau)\lceil\frac{n'-2\tau}{2(\lceil\frac{D}{2}\rceil-\tau)}\rceil$.

Therefore,

$$2\tau+ 2(\lceil\frac{D}{2}\rceil-\tau)\frac{n'-2\tau}{2(\lceil\frac{D}{2}\rceil-\tau)} \le  n  \le  2\tau+ 2(\lceil\frac{D}{2}\rceil-\tau)\left(\frac{n'-2\tau}{2(\lceil\frac{D}{2}\rceil-\tau)}+1\right) $$ \\
which implies, $$~~~~~~~~~n' \le  n  \le  n'+ 2(\lceil\frac{D}{2}\rceil-\tau)$$

Since $D <n'$, we have $n \in \Theta(n')$.

Let $x=(x_1,x_2, \cdots, x_{\lceil \frac{D}{2} \rceil-\tau})$ and $y=(y_1,y_2, \cdots, y_{\lceil \frac{D}{2} \rceil-\tau})$ be two sequences such that
$0 \le x_i \le z$, $x_i \ne 1$,  for $i=1,2, \cdots, \lceil \frac{D}{2} \rceil-\tau$ and $1 \le y_i \le z$ for $i=1,2, \cdots, \lceil \frac{D}{2} \rceil-\tau$ . We construct a tree $T_x$ from $T$ by exchanging the ports $0$ and $x_i$ at $v_{\tau+i-1}$, for $i=1,2, \cdots, \lceil \frac{D}{2} \rceil-\tau$,
and we construct a tree $T_y$ from $T$ by exchanging the ports $1$ and $y_i$ at $v_{D-\tau-i+1}$, for $i=1,2, \cdots, \lceil \frac{D}{2} \rceil-\tau$ . Let $\cT_X$ be the set of all such trees $T_x$ constructed from $T$ and let $\cT_Y$ be the set of all trees $T_y$ constructed from $T$. Then $|\cT_X|= |\cT_Y|=z^{(\lceil \frac{D}{2} \rceil-\tau)}$.
Let $\cT =\cT_X \cup \cT_Y$.
\begin{figure}[h]
\centering
\includegraphics[width=0.6\textwidth]{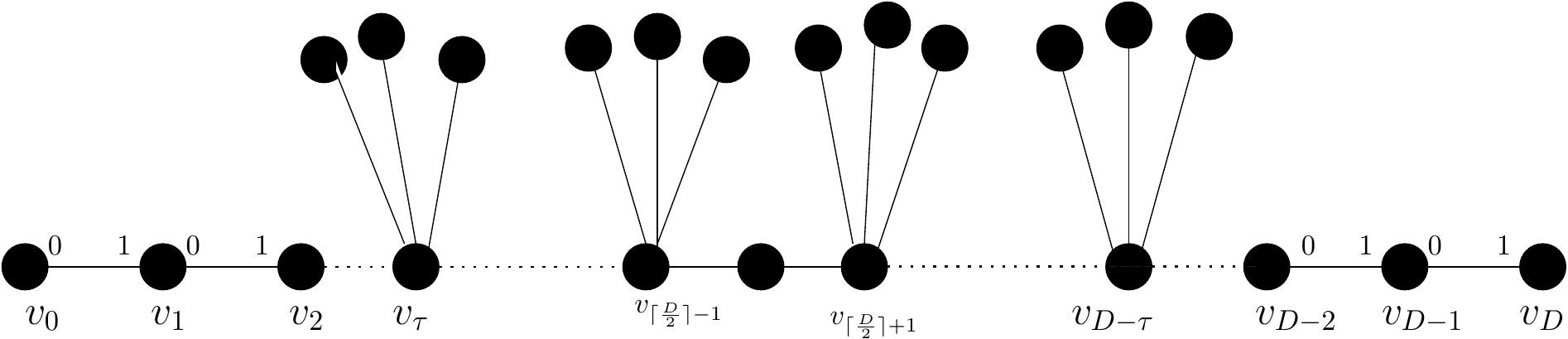}
\caption{Construction of $T$ for $D$ even}
\label{T1}
\end{figure}

\begin{figure}[h]
\centering
\includegraphics[width=0.6\textwidth]{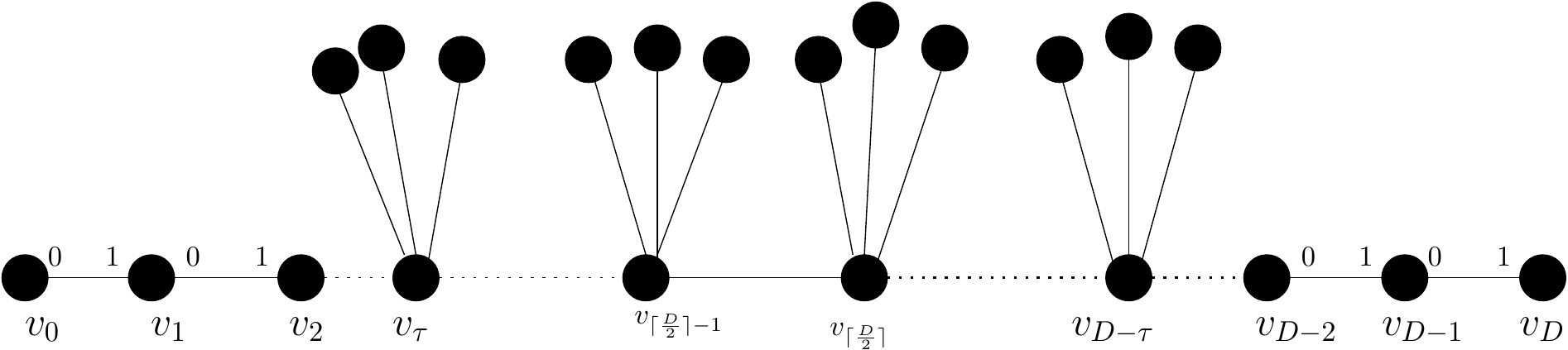}
\caption{Construction of $T$ for $D$ odd}
\label{T2}
\end{figure}

The following theorem gives a lower bound on the size of advice sufficient to perform election in time $\tau<\lceil \frac{D}{2} \rceil$, using the class of trees $\cT$
constructed above. This bound matches the upper bound from Theorem \ref{th:ub}.

\begin{theorem}\label{lem:lem1}
Consider any algorithm ELECT which solves election in $0\leq \tau<\lceil \frac{D}{2} \rceil$ rounds, for every  tree. For all integers $n' >D\ge 3$, there exists a
tree $T\in \cT$ with $n \in \Theta(n')$ nodes and diameter $D$,  for which algorithm ELECT requires
advice of size $\Omega(D\log \frac{n}{D})$, when $\tau=0$, and advice of size $\Omega \left(\frac{D-2\tau}{\tau} \log \left( \frac{n-2\tau}{D-2\tau}\right)\right)$,
when $0< \tau<\lceil \frac{D}{2} \rceil$.
\end{theorem}
\begin{proof}
We prove the theorem by contradiction. Consider an algorithm ELECT that solves election in $\tau$ rounds with advice of size $p <\left(\frac{\lceil\frac{D}{2}\rceil-\tau}{\tau+1} \log \left( \frac{n'-2\tau}{2(\lceil\frac{D}{2}\rceil-\tau)} \right)\right) -1$. Consider the execution of algorithm ELECT for the trees in $\cT$.
 For any choice of the leader,  at least one of the nodes $v_0$ or $v_D$ must be at distance at least $ \lceil \frac{D}{2}\rceil$ from it. Without loss of generality, let the distance from $v_0$ to the leader be at least $ \lceil \frac{D}{2}\rceil$.

With the size of advice at most $p$, there are at most $2^{(p+1)(\tau+1)}$ possible labeled balls $\cB (v_0,\tau)$.  Hence, the number of different pieces of information that $v_0$
can get within time $\tau$ is at most
 $2^{(p+1)(\tau+1)} < \left( \frac{n-2\tau}{2(\lceil\frac{D}{2}\rceil-\tau)}\right)^{\lceil \frac{D}{2}\rceil-\tau}\le \left\lceil \frac{n-2\tau}{2(\lceil\frac{D}{2}\rceil-\tau)}\right\rceil^{\lceil \frac{D}{2}\rceil-\tau}= |\cT_X|$. Therefore, there exist at least two trees $T_1$, $T_2$ $\in \cT_X$ such that the nodes $v_0$ in $T_1$ and $v_0$ in $T_2$ see the same labeled balls. Hence, $v_0$ in $T_1$ and $v_0$ in $T_2$ must output the same sequence of port numbers to give the path to the leader. According to the construction of the trees in $\cT_X$, for every two such trees,  the paths from $v_0$ of length at least $\lceil \frac{D}{2}\rceil$ must correspond to different sequences of port numbers. This contradicts the correctness of the algorithm ELECT. Therefore, the size of the advice must be in $\Omega\left(\frac{\lceil\frac{D}{2}\rceil-\tau}{\tau+1} \log \left( \frac{n'-2\tau}{2(\lceil\frac{D}{2}\rceil-\tau)} \right)\right)$, i.e., it is in $\Omega(D\log \frac{n}{D})$, when $\tau=0$, and it is in $\Omega \left(\frac{D-2\tau}{\tau} \log \left( \frac{n-2\tau}{D-2\tau}\right)\right)$,
when $0< \tau<\lceil \frac{D}{2} \rceil$.
\end{proof}

Theorems \ref{th:ub} and \ref{lem:lem1} imply the following corollary.

\begin{corollary}
The minimum size of advice sufficient to perform leader election in time $0\leq \tau<\lceil \frac{D}{2} \rceil$ in all $n$-node trees of diameter $D$ is
$\Theta(D\log \frac{n}{D})$, when $\tau=0$, and it is $\Theta \left(\frac{D-2\tau}{\tau} \log \left( \frac{n-2\tau}{D-2\tau}\right)\right)$,
when $0< \tau<\lceil \frac{D}{2} \rceil$.
\end{corollary}

\section{ Election with advice of constant valency}

In this section we study the minimum size of advice to perform election in time $\tau=\lfloor\alpha D \rfloor$, where $\alpha<1/2$ is a positive constant,  assuming that the advice is of constant valency $\lambda >1$.
The section is organized as follows. We first give a general construction of a class $\cT$ of trees that will be used to prove our lower bounds on the size of advice.
The rest of the section is divided into three parts, corresponding, respectively, to the cases of small, medium and large diameter $D$ of the tree, with
respect to its size $n$.  In each part we give a lower bound, using a particular case of our general construction.

In all parts, the proof of the lower bound is split into two facts. The first fact gives the requirement of the minimum size of advice that any leader election algorithm  for trees in $\cT$, working in time $\tau$, has to satisfy. The second fact establishes the condition $\xi_{\lambda}(T_1) \le \tau$ for any tree $T_1 \in  \cT$. This condition asserts that the time $\tau$ is sufficient to elect the leader in any tree from $ \cT$, if sufficiently large $\lambda$-valent advice is given to the nodes of the tree.

As for the upper bounds,
in the cases of small and of medium diameter, we use the previously mentioned upper bound $O(n)$ on the size of advice sufficient for leader election.
For small diameter this leaves a sub-polynomial gap in advice size, and for medium diameter it is tight. In the case of large diameter, we provide an election algorithm using advice of constant size, whenever the allocated time is sufficiently large.

\subsection{Construction of trees for lower bounds}\label{sec:Constant}

Let $\alpha < \frac{1}{2}$  be a positive real constant and $\lambda>1$ a positive integer constant. Let $D<n'$ be positive integers.
Let $\tau=\lfloor \alpha D \rfloor$.
We first suppose that $D$ is even. Later we will address the case when $D$ is odd.

We use the variables $k_1$, $k_2$, $z$, $z'$ in the following construction of a tree $T$.
We assume that $k_2$ is even. The values of these variables will be specified later to obtain our lower bounds for various ranges of the diameter $D$.

Let $T'$ be the tree consisting of a central node $r$ and $k_1$ subtrees $S_1$, $S_2$, $\cdots$, $S_{k_1}$  with $r$ as a common endpoint,  see Fig. \ref{T'}. For $i=1$ to $k_1$,
the subtree $S_i$ consists of $k_2$ paths $P_i^1$, $P_i^2$, $\cdots$, $P_i^{k_2}$ of length $\frac{D}{2}$ with $r$ as a common endpoint. For $i=1,2,\cdots,k_1$ and $j=1,2,\cdots, k_2$, let $v_i^j(0)$, $v_i^j(1)$, $\cdots$, $v_i^j({\frac{D}{2}-1})$ be the nodes on $P_i^j$, where $v_i^j(0)$ is the endpoint of $P_i^j$ other than $r$, and with port numbers 0 and 1 at each edge of $P_i^j$,  from $v_i^j(0)$ to $r$.
 \begin{figure}[h]
\centering
\includegraphics[width=0.5\textwidth]{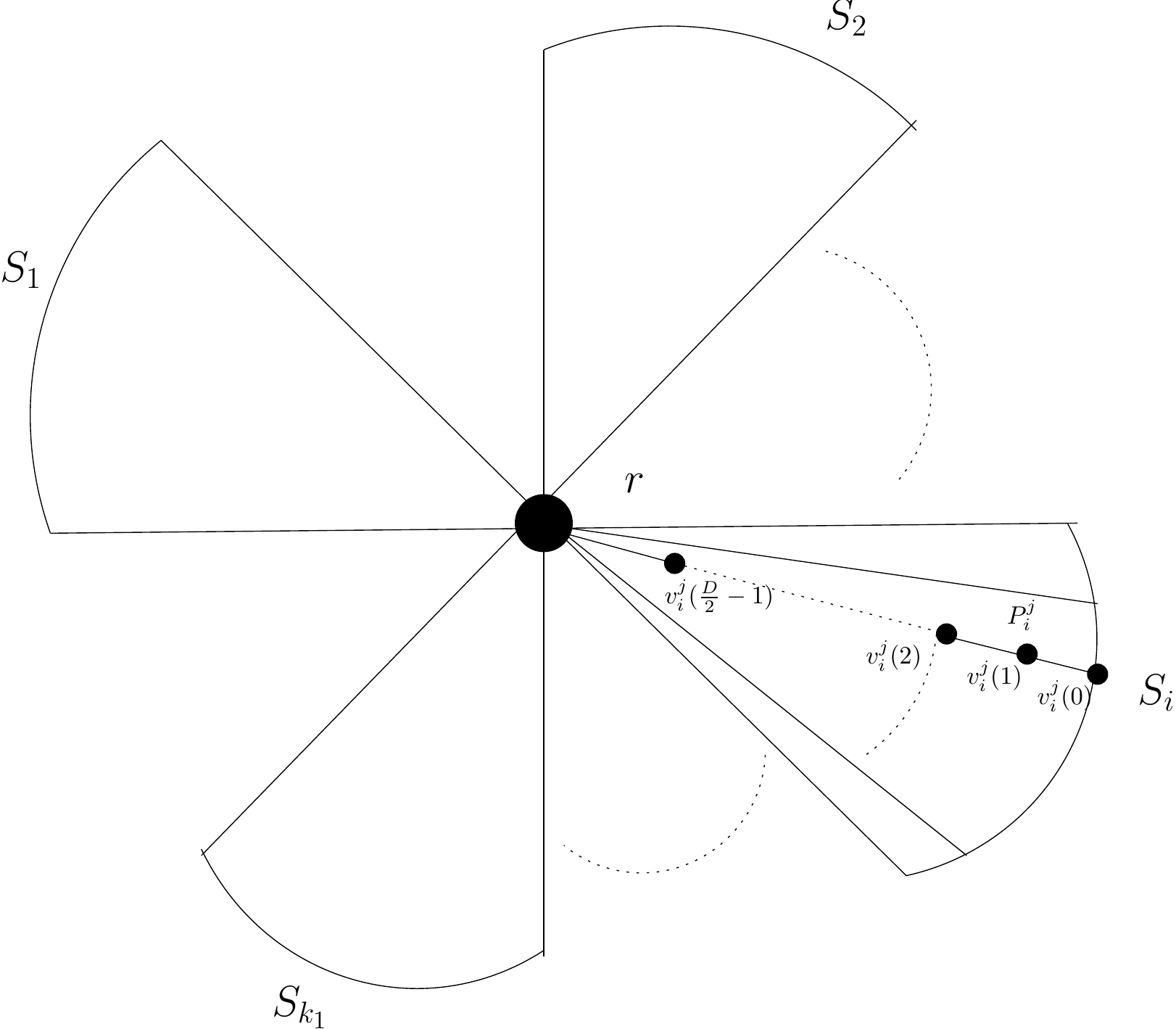}
\caption{The tree $T'$}
\label{T'}
\end{figure}

  A tree $T$ is constructed from $T'$ by attaching some leaves to the nodes of $T'$ in the following way.

  \begin{enumerate}[(1)]
  \item For each $i,j,k$, such that $1 \le i\le k_1$, $1\le j\le k_2$, and $\tau+1\le k\le \frac{D}{2}-1$,

  attach $z-1$ leaves to each of the nodes $v_i^j(k)$.
 These nodes are called {\em white}. The port numbers corresponding to the newly added edges at the node $v_i^j(k)$ are $2,\dots, z$.

 \item For each $i,j$, such that $1 \le i\le k_1$, $1\le j\le k_2$, attach $\lfloor \log D \rfloor$ additional leaves  to each of the nodes $v_i^j(k)$, for $k= q(\tau-2)+3$ for $q\ge 1$. These nodes are called {\em grey}. The port numbers corresponding to these additional edges at the node $v_i^j(k)$ are $z+1,z+2,\cdots,z+\lfloor\log D\rfloor$.

 \item For each $i,j$, such that $1 \le i\le k_1$, $1\le j\le k_2$, attach  $i-1$ additional leaves to each of the nodes $v_i^j(k)$, for $k= q(\tau-2)+2$ for $q\ge 1$. These nodes are called {\em black}.
 The port numbers corresponding to these additional edges are $z+1,z+2,\cdots,z+i-1$ at $v_i^j(k)$, for $k= q(\tau-2)+2$ and $q> 1$, and the port numbers corresponding to these additional edges are $2,3,\cdots,i$ at $v_i^j(\tau)$.

 \item For each $i,j$, such that $1 \le i\le k_1$, $1\le j\le k_2$, attach $z'$ additional leaves to each of the nodes $v_i^j(k)$, for $k= q(\tau-2)+1$ for $q\ge 1$.
 These nodes are called {\em dotted}. The port numbers corresponding to these additional edges are $z+1,z+2,\cdots,z+z'$ at $v_i^j(k)$, for $k= q(\tau-2)+1$ and $q> 1$, and the port numbers corresponding to these additional edges are $2,3,\cdots,z'+1$ at $v_i^j(\tau-1)$.

  \end{enumerate}
Let $Q_i^j$ be the subtree which is constructed by attaching the nodes as stated above to the nodes of $P_i^j$, see Fig. \ref{fig:Ng1}.

\begin{figure}
\centering
   \begin{subfigure}[b]{0.65\textwidth}
   \includegraphics[width=1\linewidth]{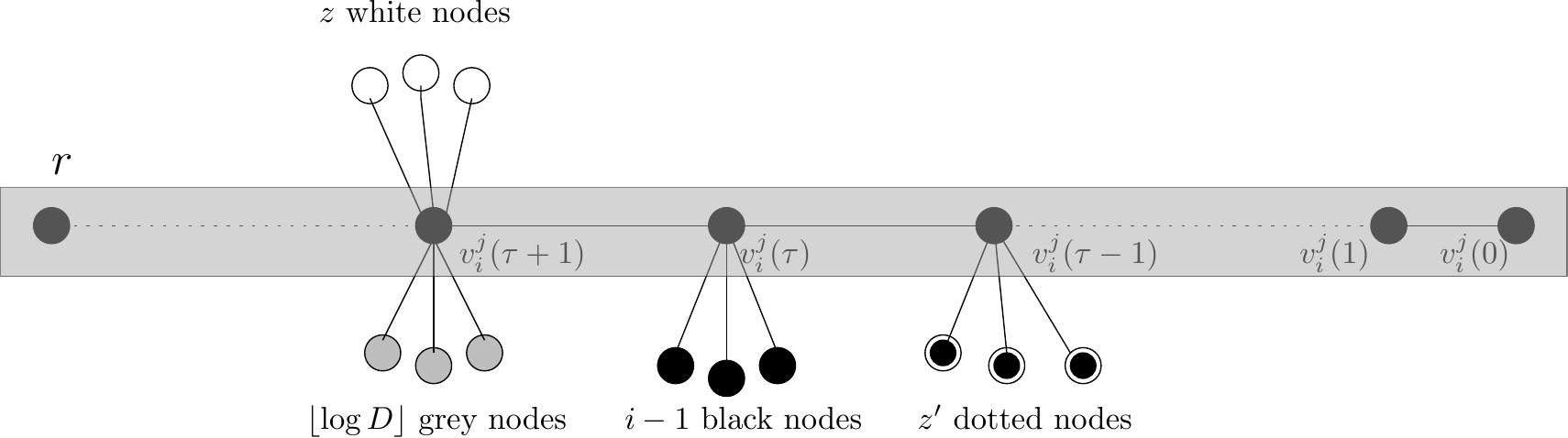}
   \caption{Nodes attached in $P_i^j$ to the nodes $v_i^j(k)$, $k \le \tau+1$. }
\end{subfigure}

\vspace*{0.5cm}

\begin{subfigure}[b]{0.65\textwidth}
   \includegraphics[width=1\linewidth]{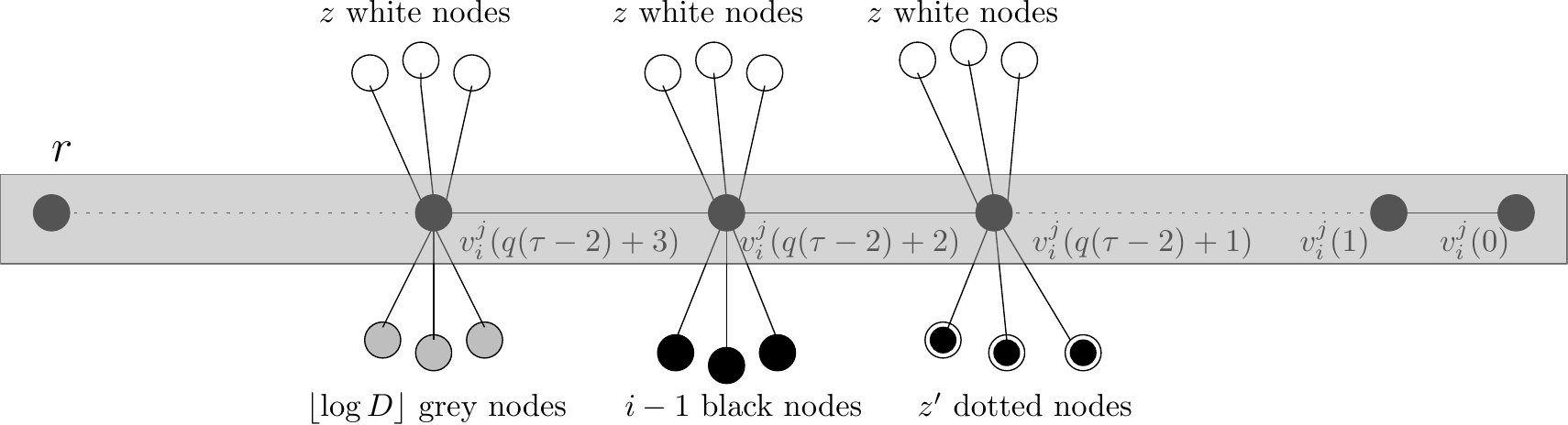}
   \caption{Nodes attached in $P_i^j$ to the nodes $v_i^j(k)$, $k > \tau+1$.}
\end{subfigure}
\caption{The construction of the tree $T$}
 \label{fig:Ng1}
\end{figure}

Let $\gamma= \lfloor \frac{D}{2(\tau-1)}\rfloor$.
The total number $n$ of nodes in the tree $T$ is given by
$$n \le k_1k_2\left( \tau +1+z(\frac{D}{2}-\tau-1)+ (z'+\lfloor \log D \rfloor+ \frac{k_1-1}{2})\gamma \right)$$.

Let $y=\frac{D}{2}-\lfloor \alpha D \rfloor-1$. Let $x=(x_i^j(k): i=1,2, \cdots,k_1, j=1,2,\cdots,{\frac{k_2}{2}},  k=1,2, \cdots,y )$, and $w=(w_i^j(k): i=1,2, \cdots,k_1, j=1,2,\cdots,{\frac{k_2}{2}},  k=1,2, \cdots,y)$ be any sequences of integers such that $0 \le x_i^j(k), w_i^j(k) \le z$ and $x_i^j(k)\neq 1$,  $w_i^j(k)\neq 1$,  for $i=1,2, \cdots,k_1 $, $j=1,2,\cdots,{\frac{k_2}{2}}$ and $k=1,2, \cdots,y $.

The tree $T_x$ is constructed from $T$ by exchanging the port numbers $x_i^j(k)$ and $0$ at the node $v_i^j(\tau+k)$ for $i=1,2, \cdots,k_1 $, $j=1,2,\cdots,{\frac{k_2}{2}}$ and $k=1,2, \cdots,y $.

 The tree $T_{w}$ is constructed from $T$ by exchanging the port numbers $w_i^j(k)$ and $0$ at the node $v_{i}^j(\tau+k)$ for $i=1,2, \cdots,k_1 $, $j=\frac{k_2}{2}+1,\frac{k_2}{2}+2,\cdots,k_2$ and $k=1,2, \cdots,y $.

 Let $\cT_X$ be the set of all such trees $T_x$ and let $\cT_W$ be the set of all such trees $T_w$. Then $|\cT_X|=|\cT_W|=z^{\frac{k_1k_2y}{2}}$.
 Let $\cT= \cT_X \cup \cT_W$.

When $D$ is odd, the tree $T$ is obtained by the same construction for $D-1$,  adding an extra edge to  $P_1^1$ in the construction of $T'$.
The rest is the same as above.

\subsection{Small diameter: $D \in \omega(1)$ and $D \in o(\log n)$}\label{small}

Let $0 <\epsilon<\frac{1}{4} $ be a real constant. Let $D$ and $n'$ be positive integers such that $D \in \omega(1)$ and $D \in o(\log n')$.
In order to prove the lower bound on the size of advice for small diameter,
we use the construction of the class  $\cT$ of trees from Section \ref{sec:Constant} for $k_1=\lceil n'^{\epsilon}\rceil$, $k_2=2\lceil\frac{ n'^{1-4\epsilon}}{2}\rceil$, $z=\lceil\frac{2(n'-n'^{1-\epsilon})}{n'^{1-3\epsilon}(1-2\alpha)D}\rceil$, and $z'=\lfloor(1-3\epsilon)\log n'\rfloor$.

The total number of nodes in a tree from $ \cT$ is
\begin{eqnarray*}
n & \leq & k_1k_2\left( \tau +1+zy+ (z'+\lfloor \log D \rfloor+\frac{k_1-1}{2})\gamma\right) \nonumber \\
& = &2\left\lceil \frac{n'^{1-4\epsilon}}{2}\right\rceil \lceil n'^{\epsilon}\rceil\left( \tau +1+\left\lceil\frac{2(n'-n'^{1-\epsilon})}{n'^{1-3\epsilon}(1-2\alpha)D}\right\rceil y+ (\lceil(1-3\epsilon)\log n'\rceil+\lfloor \log D \rfloor)\gamma+\gamma\left\lceil\frac{n'^{\epsilon}-1}{2}\right\rceil\right)
\end{eqnarray*}
Since $\tau=\lfloor \alpha D \rfloor$, $D\in o(\log n')$ and $\gamma$ is constant, we have $n \in \Theta(n')$.

Before formulating our lower bound, we explain the intuitive role of each node in a tree in $\cT$. There are four types of nodes attached to nodes of $P_i^j$ in the construction of the tree $T$. The nodes of each type have a different role in proving the lower bound in this section. The proof of our lower bound is split into two lemmas. The first lemma gives the minimum size of advice that any leader election algorithm  for trees in $\cT$, working in time $\tau$, has to satisfy. The second lemma establishes the condition $\xi_{\lambda}(T_1) \le \tau$ for any tree $T_1 \in  \cT$.

In the construction of the tree $T$, $z$ white nodes are attached to all the nodes $v_i^j(k)$, for $i=1,2, \cdots,k_1$, $j=1,2,\cdots,k_2$ and $k=\tau+1,\tau+2,\cdots,\frac{D}{2}-1$. These nodes are added, so that port number variation at nodes $v_i^j(k)$ can make the class $\cT$ sufficiently large.
All the other three types of nodes attached are used to prove the second lemma: every node must be able to identify its position in the tree $T_1\in \cT$, given sufficiently large $\lambda$-valent advice. To identify its position in $T_1$, every node must identify the integers $i$, $j$, such that it belongs to the subtree $Q_i^j$, and must identify its position in $Q_i^j$. The $i-1$ black nodes that are attached to the  nodes $v_i^j(k)$, for $k= q(\tau-2)+2$,  $q\ge 1$, help every node to identify the integer $i$. The $z'$ dotted nodes which are attached to the  nodes $v_i^j(k)$, for $k= q(\tau-2)+1$, $q\ge 1$, help every node to identify the integer $j$.
The $\lfloor \log D \rfloor$ grey nodes which are attached to the  nodes $v_i^j(k)$, for $k= q(\tau-2)+3$, $q\ge 1$, help every node to identify its distance from the node $r$,
i.e., its position in $Q_i^j$.

The next lemma gives the lower bound on the size of advice sufficient for leader election,  using the class $\cT$ of trees.

\begin{lemma}\label{lem:lb}
Let $\alpha < \frac{1}{2}$, $\delta<1$ be  positive real constants and $\lambda>1$ an integer constant. Let $D$ and $n'$ be positive integers such that $D \in \omega(1)$ and $D \in o(\log n')$.
Consider any algorithm ELECT which solves election in $\tau=\lfloor \alpha D \rfloor$ rounds with $\lambda$-valent advice, for every  tree with $\lambda$-election index at most $\tau$.
There exists
a tree $T \in \cT$  for which algorithm ELECT with $\lambda$-valent advice, working in time $\tau$, requires advice of size $\Omega(n^\delta)$.
\end{lemma}
\begin{proof}
We prove the lemma by contradiction. It is enough to prove the lemma for sufficiently large $n'$, and for $1/4 <\delta<1$.
We assume that $D$ is even. The proof for odd $D$ is similar. Fix $\epsilon =\frac{1-\delta}{3}$. It is enough to prove the lemma for $D \ge \frac{2(1-\epsilon)\log \lambda}{\epsilon(1-2\alpha)}$.
Consider an algorithm ELECT that solves election in $\tau=\lfloor \alpha D \rfloor$ rounds with advice of size $p <\frac{\frac{1-\delta}{3} n'^\delta}{\lambda}-1=\frac{\epsilon n'^{1-3\epsilon}}{\lambda}-1$.

Consider the execution of algorithm ELECT for the trees in $\cT$.
Algorithm ELECT chooses the leader either in some $Q_i^j$ for $1 \le i\le k_1$, $1\le j \le \frac{k_2}{2}$ or in some $Q_i^j$, for
$1 \le i\le k_1$, $\frac{k_2}{2}+1\le j \le k_2$. Without loss of generality suppose that the leader is chosen in some $Q_i^j$ for $\frac{k_2}{2}+1\le j \le k$. Therefore, the distance from the leader to the node $v_i^j(0)$ for $1 \le i\le k_1$, $1 \le j \le \frac{k_2}{2}$ is at least $\frac{D}{2}$.

Now, $z=\lceil \frac{2(n'-n'^{1-\epsilon})}{n'^{1-3\epsilon}(1-2\alpha)D}\rceil \ge  \frac{2n'^{3\epsilon}}{(1-2\alpha )D}- \frac{2n'^{2\epsilon}}{(1-2\alpha) D} > \frac{n'^{3\epsilon}}{(1-2\alpha) D}$.

Therefore, $\log z > 3 \epsilon \log n' -\log\left((1-2\alpha)D\right)> 2\epsilon \log n'$.

We have
$$|\cT_X|= z^{\frac{k_1k_2y}{2}}= z^{(\frac{D}{2}-\tau-1)\frac{k_1k_2}{2}}= z^{(\frac{D}{2}-\lfloor \alpha D \rfloor-1)\frac{k_1k_2}{2} } \ge$$
$$ z^{(\frac{D}{2}- \alpha D-1)\frac{k_1k_2}{2} }\geq 2^{(\frac{D}{4}(1-2\alpha)-\frac{1}{2}) n'^{1-3\epsilon} \log z} > 2^{\frac{\epsilon D}{2}(1-2\alpha) n'^{1-3\epsilon}  \log n'}> 2^{(1-\epsilon)n'^{1-3\epsilon}\log n' \log \lambda}.$$

Let $\cB(S)=(\cB (v_i^j(0),\tau): i=1,2,\cdots,k_1, j=1,2,\cdots,\frac{k_2}{2})$ be the ordered collection of all labeled balls $\cB (v_i^j(0),\tau)$, for $i=1,2,\cdots,k_1$, $j=1,2,\cdots,\frac{k_2}{2}$, in the tree $S$ from the class $\cT _X$.
With the size of advice at most $p$, there are at most $2^{p+1}$ possible advice strings. Hence there are at most ${2^{p+1} \choose {\lambda}} \leq 2^{(p+1)\lambda} $ choices of $\lambda$ such strings, and thus
there are at most $2^{(p+1)\lambda} {\lambda^{(\tau+1+z')\frac{k_1k_2}{2}}}$ possible sequences $\cB(S)$ because the size of the ball $\cB (v_i^j(0),\tau)$ is
$\tau+1 +z'$.

Now,
\begin{eqnarray*}
2^{(p+1)\lambda} \lambda^{(\tau+1+z') \frac{k_1k_2}{2}} &=& 2^{(p+1)\lambda} \lambda^{(\lfloor \alpha D \rfloor +1+z') \frac{k_1k_2}{2}} \nonumber \\
%&\le & 2^{(p+1)\lambda} \lambda^{(\alpha D +1+z') \frac{k_1k_2}{2}} \nonumber \\
&\le &  2^{\left(\epsilon n'^{1-3\epsilon}+(\alpha D +1+ (1-3\epsilon) \log n') (\frac{n'^{1-3\epsilon}+n'^{1-4\epsilon}+2n'^{\epsilon}+2}{2})\log \lambda\right)} \nonumber \\
& < & 2^{ \left(\epsilon n'^{1-3\epsilon}+(1-2\epsilon) n'^{1-3\epsilon} \log n' \log \lambda)\right)} \nonumber \\
& < & 2^{(1-\epsilon)n'^{1-3\epsilon}\log n' \log \lambda} \nonumber \\
& < &  |\cT_X|
\end{eqnarray*}
 Hence, there exist at least two trees $T_1, T_2 \in \cT_X$ such that
 $\cB(T_1)=\cB(T_2)$. Therefore the nodes $v_i^j(0)$ in $T_1$ and $v_i^j(0)$ in $T_2$, for $i=1,2,\cdots,k_1$, $i=1,2,\cdots,\frac{k_2}{2}$, must output the same sequence of port numbers to give the path to the leader. According to the construction of the trees in $\cT_X$, there exists a node $v_i^j(0)$, for some $1 \leq i \leq k_1$ and $1 \le j\le \frac{k_2}{2}$, such that the path to the leader from  $v_i^j(0)$ in $T_1$ and the path to the leader from  $v_i^j(0)$ in $T_2$ correspond to different sequences of port numbers. This is a contradiction. Therefore, the size of the advice is $\Omega(n^{\delta})$.
\end{proof}

The next lemma shows that the $\lambda$-election index of trees in $\cT$ does not exceed $\tau$.

\begin{lemma}\label{lem:xi}
For any tree $T_1 \in \cT$, $\xi_{\lambda}(T_1) \le \tau$.
\end{lemma}
\begin{proof}
We assume that $D$ is even. The proof for odd $D$ is similar.
Let $\tau= \lfloor \alpha D\rfloor$. Since $D\in \omega(1)$, we may assume that $\tau>2$.
In order to prove the lemma, we present a leader election algorithm  working in time $\tau$, if some $\lambda$-valent advice of sufficient size is available to the nodes.
For any tree in $\cT$, the node $r$ is chosen as the leader.

At a high level,  we assign to each subtree $Q_i^j$ of $S_i$ a different string of length $\tau -2$ with at most  $\lambda$ terms, called  {\em colors}. The pieces of advice at the $z'$ dotted nodes attached to each node $v_i^j(q(\tau-2)+1)$, for $q\ge 1$, form such a string.
Since $k_2=\lfloor n'^{1-4\epsilon}\rfloor$ and $z'=\lfloor(1-3\epsilon)\log n'\rfloor$, therefore, such a one-to-one assignment is possible.
The degree of the node $v_i^j(q(\tau-2)+2)$ is $z+i$ for $q>1$, and is $i+1$ for $q=1$.  In time $\tau$, every node can see at least one node $v_i^j(q(\tau-2)+2)$, for $q\ge 1$. It knows $z$, hence it can learn $i$, and thus can identify the subtree $S_i$ to which it belongs.
The grey nodes attached to the nodes $v_i^j(k)$ of $P_i^j$ are used to identify the distance from $v_i^j(k)$ to $r$.
Knowing $i$,  every node identifies the subtree $Q_i^j$ to which it belongs, by computing the unique string associated with $Q_i^j$ in the subtree $S_i$, in time $\tau$. The node identifies its position in $Q_i^j$ by either seeing the endpoint $v_i^j(0)$ or by seeing two nodes of degree $
z+\lfloor \log D \rfloor$, together with their neighbors.

We now describe formally the advice assignment to the nodes of a tree $T_1 \in \cT$.
Let $\cC=\{c_1,c_2,\cdots, c_\lambda\}$ be a set of $\lambda$ colors. Let $m=\lambda^{\tau-2}$ and let $S=\{s_1,s_2,\cdots, s_m\}$ be the set of sequences of colors of length $\tau-2$. Let $s_i(j)$ be the $j-$th term of $s_i$. Let $V(T_1)$ be the set of nodes of $T_1$. Define $f: V(T_1) \rightarrow \cC$ as follows. Assign $f(r)=c_1$. We divide all the nodes in $Q_i^j$, for $i=1,2,\cdots,k_1$ and $j=1,2,\cdots,k_2$, into the following four types, and assign colors to them as follows.

\noindent
{\bf Type 1:} All nodes $v_i^j(k) \in Q_i^j$, $0 \le k \le \frac{D}{2}-1$ are  of this type. Assign   $f(v_i^j(k))=c_1$ for $0 \le k \le \frac{D}{2}-1$.

\noindent
{\bf Type 2:} All grey nodes are of this type. Let $w_\ell$, for $\ell=1,\dots \lfloor \log D \rfloor$,  be the
grey node attached to its only neighbor using port $z+\ell$ at this neighbor. Assign $f(w_{\ell})\in\{c_1,c_2\}$, so that the sequence $(f(w_1)f(w_2) \cdots f(w_{\lfloor\log D\rfloor})$ be the binary representation of the distance from $r$ to the only neighbor of $w_{\ell}$, with $c_1$ standing for 0 and $c_2$ standing for 1.

\noindent
{\bf Type 3:} All dotted nodes are of this type. Let
$u_\ell$, for $\ell=1,\dots z'$,  be the dotted
node in $Q_i^j$ attached to its only neighbor using port $z+\ell$ at this neighbor. Assign $f(u_{\ell})\in\{c_1,c_2,\cdots,c_\lambda\}$, so that  $s_j=(f(u_1)f(u_2) \cdots f(u_{z'}))$.

\noindent
{\bf Type 4:} All white and black nodes are of this type.  Assign $f(v)=c_1$ for all these nodes.

For all other nodes $v\in Q_i^j$, assign $f(v)=c_1$.

 Let $T_1(f)$ be the  node-colored map of $T_1$ corresponding to the color assignment $f$.
The advice provided to each node $v$ in $T_1$ is $(f(v),T_1(f))$.
We show that each node $v$ in $T_1$ can identify itself in $T_1(f)$ in time $\tau$, using advice $(f(v),T_1(f))$. This is enough to output the sequence of port numbers corresponding to the path to the leader.
Consider a node $v$ in $T_1$.
\begin{itemize}

\item [Case 1:]  $v=r$.

The node $r$ can identify itself in time 1 (without using any advice) as the unique node of degree larger than 2 all of whose neighbors have degree larger than 2.

\item [Case 2:]  $v \in Q_i^j$ is of Type 1.

If $v=v_i^j(0)$, then it can identify itself as one of the endpoints of some $P_i^j$, for $1\le i \le k_1$ and $1\le j \le k_2$, as every node $v_i^j(0)$ can see a line of length $\tau-2$ with itself as one end point.
     The labeled ball $\cB(v^j_i(0),\tau)$ is a labeled subtree of diameter $\tau$ with $\tau+z'+1$ nodes.
     $v_i^j(0)$ identifies the integer $i$ by looking at the degree of the only non-leaf at distance $\tau$. For example, if the degree of this node is 5, then the node computes $i$ as 4, because according to the construction of the tree $T$, the degree of the node $v_i^j(\tau)$ is $i+1$. Therefore, the node $v_i^j(0)$ identifies the subtree $S_i$ to which it belongs.
     The node $v_i^j(0)$ computes the string of length $\tau-2$ by looking at the colors of the dotted nodes attached to the only node of degree $z'+2$ in $\cB(v_i^j(0),\tau)$.
    Since this string of length $\tau-2$ uniquely determines $Q_i^j$ in $S_i$, node $v_i^j(0)$ identifies itself as a node in $Q_i^j$.

    If $v=v_i^j(k)$, for $ 0 <k \le \tau$, then it can identify itself as a node in some $P_i^j$, $1\le i \le k_1$ and $1\le j \le k_2$. The node $v_i^j(0)$ is  in $\cB(v,\tau)$,  and hence $v$ learns the distance to $v_i^j(0)$ by seeing $\cB(v,\tau)$. Then the node $v$ computes the integer $i$ and the string of colors  similarly as stated for $v_i^j(0)$, identifies the subtree $Q_i^j$ in $S_i$ to which it belongs, and identifies its position in $Q_i^j$.

    If $v=v_i^j(k)$, for $ \tau+1 \le k \le \frac{D}{2}-1$, then the node $v$ can see five kinds of nodes in $\cB(v,\tau)$:  nodes of degree one, nodes of degree $z+1$, at least two  nodes of degree $z+\lfloor \log D \rfloor$, at least one  node of degree $z+z'$ and at least one node of degree $z+i$. Since the colored map $T_1(f)$ is a part of the advice given to every node, the nodes can distinguish between the above kinds of nodes. The node identifies $i$ and thus $S_i$ to which it belongs, since it knows $z$.
 Then it computes the string $s$ of length $\tau-2$ by collecting the colors from the leaves attached to a node with degree $z+z'$. Hence, it identifies $j$, and thus identifies the subtree $Q_i^j$ in the class $S_i$ to which it belongs.

    The node $v$ sees at least two nodes $u_1,u_2$ in $\cB(v,\tau)$ with degree $z+\lfloor \log D \rfloor$, and all the neighbors of $u_1$ and $u_2$ are also in $\cB(v,\tau)$.
    The node $v$ first computes the binary strings $s'$ and $s''$ from the advice at the leaves attached to $u_1$ and $u_2$, respectively, by ports $z+1,z+2, \cdots z+\lfloor \log D \rfloor$. Without loss of generality, let  $s'$ represent the integer which is larger than the integer represented by $s''$.
 Recall that these integers are the distances from the respective nodes to $r$. Hence the node $v$ can identify the direction towards $r$ which is from $u_1$ to $u_2$ along $P_i^j$,  and it can identify its position in $Q_i^j$.

 \item [Case 3:] $v \in Q^i$ is of Type 2 or of Type 3 or of  Type 4.

 The node $v$ follows similar steps as in {Case 2}, where $v=v_i^j(k)$, for $ \tau+1 \le k \le \frac{D}{2}-1$, to identify the subtree $Q_i^j$ to which it belongs, and
 to identify its position in $Q_i^j$.
\end{itemize}
\end{proof}

Lemmas \ref{lem:lb} and \ref{lem:xi} imply the following result.

\begin{theorem}
Let $\alpha < \frac{1}{2}$ and $\delta<1$ be positive real constants.  Let $\lambda>1$ be an integer constant. Let $D$ and $n'$ be positive integers such that $D \in \omega(1)$ and $D \in o(\log n')$.
Consider any algorithm ELECT which solves election in $\tau=\lfloor \alpha D \rfloor$ rounds with $\lambda$-valent advice, for every  tree with $\lambda$-election index at most $\tau$.
There exists
an $n$-node tree $T_1$, where $n \in \Theta(n')$,   with diameter $D$ and  $\xi _{\lambda}(T_1) \leq \tau$,  for which  algorithm ELECT with $\lambda$-valent advice, working in time $\tau$, requires advice of size $\Omega(n^{\delta})$.
\end{theorem}

\subsection{Medium diameter: $D \in \omega(\log n)$ and $D \in o(n)$}

 Let $D$ and $n'$ be positive integers such that $D \in \omega(\log n')$ and $D \in o(n')$.
%We first suppose that $D$ is even. Later we will address the case when $D$ is odd.
Let $b$ be a positive real constant, such that $\frac{1}{b} > 1+ 4(1-2\alpha)\lambda^{\frac{4(1+4\alpha)}{1-2\alpha}}$ and $\frac{bn'}{D}>1$.
In order to prove our lower bound on the size of advice,
we now use the construction of the class $\cT$ of trees in Section \ref{sec:Constant} for $k_1=1$, $k_2=k=2\lceil\frac{bn'}{D}\rceil$, $z=\lceil\frac{n'-bn'}{k(\frac{D}{2}-\lfloor \alpha D \rfloor-1)}\rceil$, and $z'=0$.

The total number of nodes in a tree from $ \cT$ is
\begin{eqnarray*}
n & \leq & k \left( \tau+1+\left\lceil\frac{n'-bn'}{k(\frac{D}{2}-\tau-1)}\right\rceil\left(\frac{D}{2}-\tau-1\right)+ \lfloor \log D \rfloor \gamma \right)
 \nonumber \\
\end{eqnarray*}
Since $\tau=\lfloor \alpha D \rfloor$, $D\in o( n')$ and $\gamma$ is constant, we have $n \in \Theta(n')$.

Since $k_1=1$, in what follows, we omit the running index ranging from 1 to $ k_1$. In particular, $P_i^j$, $Q_i^j$ and $v_i^j(l)$ are replaced, respectively,
by $P^j$, $Q^j$ and $v^j(l)$.

As in Section \ref{small}, we explain the role of each node in $T$ in this case. The role of the white nodes is the same as before, i.e.,
these nodes are added, so that port number variation at nodes $v^j(l)$, for $j=1,2,\cdots,k$ and $l=\tau+1,\tau+2,\cdots,\frac{D}{2}-1$, can make the class $\cT$ sufficiently large. The nodes on $P^j$ are used
to assign a different string of colors to each $Q^j$. Since $D \in \omega(\log n)$ in our present case, there are enough such nodes, as opposed to the situation of small $D$, when dotted nodes had to be added for this purpose.
The $\lfloor \log D \rfloor$ grey nodes which are attached to the  nodes $v^j(l)$, for $l= q(\tau-2)+3$, $q\ge 1$, are there to help every node to identify its distance from the node $r$,
i.e., its position in $Q^j$.

As before, the proof of the lower bound is split into two lemmas, concerning, respectively, the size of advice needed for election in trees from $\cT$, and the $\lambda$-election index of these trees.

\begin{lemma}\label{lem1}
Let $\alpha < \frac{1}{2}$ be a positive real constant and $\lambda>1$ an integer constant. Let $D$ and $n'$ be positive integers such that $D \in \omega(\log n')$ and $D \in o(n')$.
Consider any algorithm ELECT which solves election in $\tau=\lfloor \alpha D \rfloor$ rounds with $\lambda$-valent advice, for every  tree with $\lambda$-election index at most $\tau$.
There exists
a tree $T \in \cT$  for which algorithm ELECT with $\lambda$-valent advice, working in time $\tau$, requires advice of size $\Omega(n)$.
\end{lemma}
\begin{proof}
We prove the lemma by contradiction.  It is enough to prove the lemma for sufficiently large $n'$. We assume that $D$ is even. The proof for odd $D$ is similar.
Consider an algorithm ELECT that solves election in $\tau=\lfloor \alpha D \rfloor$ rounds with advice of size $p <\frac{ bn'}{\lambda}-1$. Algorithm ELECT chooses the leader either in some $Q^j$ for $1\le j \le \frac{k}{2}$, or in some $Q^j$  for
$\frac{k}{2}+1\le j \le k$. Without loss of generality suppose that the leader is chosen in some $Q^j$ for $\frac{k}{2}+1\le j \le k$. Therefore, the distance from the leader to the node $v^j(0)$ for $1 \le j \le \frac{k}{2}$ is at least~$\frac{D}{2}$.

Now, $$z=\lceil\frac{n'-bn'}{k(\frac{D}{2}-\lfloor \alpha D \rfloor-1)}\rceil \ge \frac{n'-bn'}{k(\frac{D}{2}-\lfloor \alpha D \rfloor-1)}$$
$$= \frac{n'-bn'}{2\lceil \frac{bn'}{D}\rceil(\frac{D}{2}-\lfloor \alpha D \rfloor-1)}\ge \frac{n'-bn'}{ 2(\frac{bn'}{D}+1)(\frac{D}{2}-\lfloor \alpha D \rfloor-1)}$$
$$\ge \frac{n'-bn'}{ \frac{4bn'}{D}(\frac{D}{2}-\lfloor \alpha D \rfloor-1)} \ge \frac{n'-bn'}{ \frac{4bn'}{D}(\frac{D}{2}- \alpha D )}=\frac{1-b}{2b(1-2\alpha)}.$$
Also, $ky=k(\frac{D}{2}-\lfloor \alpha D \rfloor-1) = 2\lceil\frac{bn'}{D}\rceil(\frac{D}{2}-\lfloor \alpha D \rfloor)-k \ge 2\frac{bn'}{D}(\frac{D}{2}- \alpha D )-k\geq{(1-2\alpha)bn'-k}$.

 Therefore,
\begin{eqnarray*}
|\cT_X|=z^{\frac{ky}{2}} &> & \left(\frac{1-b}{2b(1-2\alpha)}\right)^{\frac{(1-2\alpha)bn'-k}{2}} \\
&  =& 2^{\left(\frac{(1-2\alpha)bn'-k}{2}\right)\log \left(\frac{1-b}{2b(1-2\alpha)}\right)} \\
&  >& 2^{\left({\frac{(1-2\alpha)bn'}{4}}\right)\log \left(\frac{1-b}{2b(1-2\alpha)}\right)}.\\
\end{eqnarray*}
The last inequality follows from  $k < \frac{(1-2\alpha)}{2}bn'$ (as $D \in \omega(\log n')$).

Let $\cB(S)=(\cB (v^1(0),t), \cB (v^2(0),t), \cdots, \cB (v^{\frac{k}{2}}(0),t))$ be the ordered collection of all labeled balls $\cB (v^i(0),t)$, for $i=1,2,\cdots,\frac{k}{2}$
for a tree $S$ from the class $\cT _X$.
With the size of advice at most $p$, there are at most $2^{p+1}$ possible advice strings. Hence there are at most ${2^{p+1} \choose {\lambda}} \leq 2^{(p+1)\lambda} $ choices of $\lambda$ such strings, and thus
there are at most $2^{(p+1)\lambda} {\lambda^{(\tau+1)\frac{k}{2}}}$ possible sequences $\cB(S)$.

Since $\frac{1}{b} > 1+ 4(1-2\alpha)\lambda^{\frac{4(1+4\alpha)}{1-2\alpha}}$, then $\frac{1-b}{4b(1-2\alpha)} > \lambda^{\frac{4(1+4\alpha)}{1-2\alpha}}$. Therefore, $\frac{1-2\alpha}{4}\log(\frac{1-b}{4b(1-2\alpha)}) > (1+4\alpha)\log \lambda$.

We have
$$2^{(p+1)\lambda} {\lambda^{(\tau+1) \frac{k}{2}}}< 2^{bn'} \lambda^{(\lfloor \alpha D \rfloor +1) \frac{k}{2}}\leq 2^{bn'} \lambda^{( \alpha D  +1) \frac{k}{2}}$$
$$<2^{bn'+(\alpha bn' +\frac{bn'}{D}+\alpha D+1)\log \lambda}< 2^{bn'+4\alpha bn' \log \lambda}\le 2^{(1+4\alpha) bn' \log \lambda}<2^{ \frac{bn'(1-2\alpha)}{4}\log(\frac{1-b}{2b(1-2\alpha)})}<|\cT_X|.$$ Hence, there exist at least two trees $T_1, T_2 \in \cT_X$ such that
 $\cB(T_1)=\cB(T_2)$. Therefore the nodes $v^j(0)$ in $T_1$ and $v^j(0)$ in $T_2$ for $j=1,2,\cdots,\frac{k}{2}$, must output the same sequence of port numbers to give the path to the leader. According to the construction of the trees in $\cT_X$, there exists a node $v^l(0)$, $1 \le l\le \frac{k}{2}$ such that the path to the leader from  $v^l(0)$ in $T_1$ and the path to the leader from  $v^l(0)$ in $T_2$ correspond to different sequences of port numbers. This is a contradiction. Therefore, the size of the advice is $\Omega(n)$.
 \end{proof}

\begin{lemma}\label{lem2}
For any tree $T_1 \in \cT$, $\xi_{\lambda}(T_1) \le \tau$.
\end{lemma}
\begin{proof}
The proof of this lemma is similar to the proof of Lemma \ref{lem:xi}. Since $k_1=1$, every node has to identify the subtree $Q^j$ to which it belongs, and its position in $Q^j$.
Since $D\in \omega(\log n)$, we may assume that $\tau>2$.
In order to prove the lemma, we present a leader election algorithm  working in time $\tau$, if some $\lambda$-valent advice of sufficient size is available to the nodes.
The node $r$ of a tree $T_1 \in \cT$ is chosen as the leader.
At a high level,  we assign to each subtree $Q^j$ a different string of length $\tau -2$ with at most  $\lambda$ colors, using the nodes in $P^j$. Since $D \in \omega(\log n)$ and $\tau=\lfloor \alpha D \rfloor$, then $k \le 2 ^{\tau}$. Therefore, such a one-to-one mapping is always possible. The grey nodes attached to the nodes $v^j(l)$ of $P^j$ are used to identify the distance from $v^j(l)$ to $r$. Every node identifies the subtree $Q^j$ to which it belongs, by computing the unique string associated with $Q^j$, in time $\tau$. The node identifies its position in $Q^j$ by either seeing the endpoint $v_j(0)$ or by seeing two nodes of degree $
z+\lfloor \log D \rfloor$, together with their grey neighbors.

As before, $m=\lambda^{\tau-2}$ and $S=\{s_1,s_2,\cdots, s_m\}$ is the set of sequences of colors of length $\tau-2$. Let $s_i(l)$ be the $l-$th term of $s_i$.
The main difference with respect to the proof of Lemma \ref{lem:xi} is the following. While in the previous proof we assigned colors $c_1,\dots,c_{\lambda}$ to dotted nodes in order to assign a different string of colors to each $Q_i^j$, we now assign the colors to the nodes on $P^j$ in the following way.
 Assign   $f(v^j(0))=f(v^j(1))=f(v^j(2))=c_1$; $f(v^j(l))= s_j\left((l-3) \mod (\tau-2)\right)$, for $ 2\le l \le \tau$. For all nodes $v^j(l) \in P^j$, $\tau+1 \le l \le \frac{D}{2}-1$, assign   $f(v^j(l))= s_j\left((l-3) \mod (\tau-2)\right)$, for $ \tau+1\le l \le \frac{D}{2}-1$.
The values of the function $f$ for all other nodes are identical as in the proof of Lemma \ref{lem:xi}.
As before,  let $T_1(f)$ be the  node-colored map of $T_1$ corresponding to the color assignment $f$.
The advice provided to each node $v$ in $T_1$ is $(f(v),T_1(f))$.

It remains to explain how a node $v$ identifies the sequence $s\in S$ corresponding to the subtree $Q^j$ to which it belongs, and its position in $Q^j$.
Node $v$ can either see the nodes $v^j(0),\dots, v^j(\tau) $ or it can see two nodes of degree $z+\lfloor \log D \rfloor$ on $P^j$, together with their grey neighbors.
In the first case, node $v$ finds the sequence $s$ of length $\tau-2$ by reading the colors  $f(v^j(3)), \dots , f(v^j(\tau))$ in this order. Since it can see the node $v^j(0)$, it can also identify its position on the map.
In the second case, node $v$ decides which of the two nodes  of degree $z+\lfloor \log D \rfloor$ is closer to $r$ by reading advice in their grey neighbors, similarly as in the proof of Lemma \ref{lem:xi}. It then finds the string $s$ by reading the colors assigned to nodes between these two nodes, from the farther to the closer.
It identifies its position with respect to the farther of them.
\end{proof}

Lemmas \ref{lem1} and \ref{lem2} imply the following theorem.

\begin{theorem}
Let $\alpha < \frac{1}{2}$ be a positive real constant and $\lambda>1$ an integer constant. Let $D$ and $n'$ be positive integers such that $D \in \omega(\log n')$ and $D \in o(n')$.
Consider any algorithm ELECT which solves election in $\tau=\lfloor \alpha D \rfloor$ rounds with $\lambda$-valent advice, for every  tree with $\lambda$-election index at most $\tau$.
There exists
an $n$-node tree $T_1$, where $n \in \Theta(n')$,   with diameter $D$ and  $\xi _{\lambda}(T_1) \leq \tau$,  for which  algorithm ELECT with $\lambda$-valent advice, working in time $\tau$, requires advice of size $\Omega(n)$.
\end{theorem}

\subsection{Large diameter: $D\in \Theta(n)$}
Let $D$ and $n'$ be positive integers such that $D =cn'+o(n')$, for some positive constant $c<1$. Let $\lambda>1$ be a constant integer.
The main result of this section gives two reals $0 < \beta _1 <\beta_2<\frac{1}{2}$, whose difference is small, which depend only on constants $c$ and $\lambda$,
and which satisfy the following properties:

\begin{enumerate}
\item
for any constant $\beta<\beta_1$, any election algorithm working in time $\tau=\lfloor\beta D \rfloor$ requires $\lambda$-valent advice of size $\Omega(n)$,
in some trees of diameter $D$ and size $n= n'+o(n')$, with $\lambda$-election index at most~$\tau$;
\item
for any constant $\beta>\beta_2$, there exists an election algorithm  working in time $\tau= \lfloor\beta D \rfloor$ with $\lambda$-valent advice of constant size, for all $n'$-node trees of diameter $D$ whose $\lambda$-election index is at most~$\tau$.
\end{enumerate}

\subsubsection{Lower bound}

The proof of the first (negative) result is split, as before, into two lemmas, concerning, respectively, the size of advice needed for election in trees from $\cT$, and the $\lambda$-election index of these trees.

\begin{lemma}\label{lem:large1}
Let $D$ and $n'$ be positive integers such that $D =cn'+o(n')$, for some positive constant $c<1$. Let $\lambda>1$ be a constant integer. There exists a real $\beta_1$, $0 < \beta _1 <\frac{1}{2}$, depending only on $c$ and $\lambda$, such that for any algorithm  ELECT which solves election in $\tau=\lfloor \beta D \rfloor$ rounds with $\lambda$-valent advice, for every  tree with $\lambda$-election index at most $\tau$, for any constant $\beta<\beta _1$,  there exists
a tree $T \in \cT$  with diameter $D$ and $n= n'+o(n')$ nodes, for which algorithm ELECT with $\lambda$-valent advice, working in time $\tau$, requires advice of size $\Omega(n)$.
\end{lemma}

\begin{proof}
We first do the proof for $D=cn'$.
It is enough to prove the lemma for $n' \geq \frac{400}{1-c}$. Let  $\epsilon = \frac{1-c}{200}$.
We show that if a real $\beta_1$ satisfies the equation $\beta_1= \frac{1-2\beta_1}{2} \log_{\lambda} \left(\frac{\frac{1}{2}-\beta_1 c}{\frac{c}{2}-\beta_1 c+\epsilon}\right)$, then any algorithm working in time $\tau=\lfloor\beta D\rfloor$, for $\beta <\beta_1$, with $\lambda$-valent advice, requires advice of size $\Omega(n)$.
 Consider an algorithm ELECT that solves election in $\tau=\lfloor \beta D \rfloor$ rounds with advice of size $p <(\beta_1 -\beta)cn'\log \lambda-\log \lambda-1$.

We use the construction of the class $\cT$ of trees in Section \ref{sec:Constant} for $k_1=1$, $k_2=2$, $z=\lceil\frac{n'-2\tau}{2(\lceil\frac{D}{2}\rceil-\tau)}\rceil$, and $z'=0$.
The total number of nodes $n$ is at most  $k_1k_2\left( \tau+1+z(\frac{D}{2}-\tau-1)+ (z'+\lfloor \log D \rfloor+ \frac{k_1-1}{2})\gamma \right) =2( \tau +1+\lceil\frac{n'-2\tau}{2(\lceil\frac{D}{2}\rceil-\tau)}\rceil (\frac{D}{2}-\tau) + \lfloor \log D \rfloor \gamma)$. This implies that $n=n'+o(n')$.
Without loss of generality, we assume that  the distance from $v^1(0)$ to the leader is at least $\lceil \frac{D}{2}\rceil$.

We have $z=   \lceil\frac{n'-2\tau}{2(\lceil\frac{D}{2}\rceil-\tau)}\rceil \ge \frac{\frac{n'}{2}-\lfloor \beta c n'\rfloor}{(\lceil\frac{cn'}{2}\rceil-\lfloor\beta cn'\rfloor)} \ge \frac{\frac{n'}{2}- \beta c n'}{(\lceil\frac{cn'}{2}\rceil-\beta cn')+1}\ge \frac{\frac{n'}{2}- \beta c n'}{(\frac{cn'}{2}-\beta cn')+2} \ge \frac{\frac{1}{2}- \beta c }{(\frac{c}{2}-\beta c)+\epsilon}$.

$|\cT_X|=z^{\lceil \frac{cn'}{2}\rceil -\lfloor c\beta n'\rfloor} \ge z^{\frac{cn'}{2} - c\beta n'}\ge {\left(\frac{\frac{1}{2}- \beta c }{(\frac{c}{2}-\beta c)+\epsilon}\right)}^{\frac{cn'}{2}-c\beta n'}=2^{\left(\frac{cn'(1-2\beta)}{2} \log (\frac{\frac{1}{2}-\beta c}{\frac{c}{2}-\beta c+\epsilon})\right)}$.

With the size of advice at most $p$, there are at most $2^{(p+1) }\lambda^{(\tau+1)}$ possible labeled balls $\cB (v_0,\tau)$.  Hence, the number of different pieces of information that $v_0$
can get within time $\tau$ is at most
$2^{(p+1) }\lambda^{(\tau+1)} < 2^{{c\beta_1 n'}\log \lambda}=2^{\left(\frac{cn'(1-2\beta_1)}{2} \log (\frac{\frac{1}{2}-\beta_1 c}{\frac{c}{2}-\beta_1 c+\epsilon})\right)}< 2^{\left(\frac{cn'(1-2\beta)}{2} \log (\frac{\frac{1}{2}-\beta c}{\frac{c}{2}-\beta c+\epsilon})\right)}$.

The last inequality follows from the fact that the function $f(\beta)=\left(\frac{(1-2\beta)}{2} \log (\frac{\frac{1}{2}-\beta c}{\frac{c}{2}-\beta c+\epsilon})\right)$ is a strictly decreasing function for $0 <\beta <\frac{1}{2}$.

Therefore,  $2^{(p+1) }\lambda^{(\tau+1)} <2^{\left(\frac{cn'(1-2\beta)}{2} \log (\frac{\frac{1}{2}-\beta c}{\frac{c}{2}-\beta c+\epsilon})\right)} =|\cT_X|$.

 Hence, there exist at least two trees $T_1$, $T_2$ $\in \cT_X$ such that the node $v^1(0)$ in $T_1$ and $v^1(0)$ in $T_2$ see the same labeled balls. Hence, $v^1(0)$ in $T_1$ and $v^1(0)$ in $T_2$ must output the same sequence of port numbers to give the path to the leader. According to the construction of the trees in $\cT_X$, for every two such trees,  the paths of length at least $\lceil \frac{D}{2}\rceil$ from $v^1(0)$ must correspond to different sequences of port numbers. This contradicts the correctness of the algorithm ELECT. Therefore, the size of the advice must be in
$\Omega(n)$.

The generalization of the reasoning to the case $D=cn'+o(n')$ follows from continuity arguments. It can be observed that the real $\beta_1$ in this case can be found
arbitrarily close to that derived for the case $D=cn'$, for sufficiently large $n'$.
\end{proof}

\begin{lemma}\label{lem:xi1}
For any tree $T_1 \in \cT$, $\xi_{\lambda}(T_1) \le \tau$.
\end{lemma}

\begin{proof}
The proof of this lemma is similar to the proof of Lemma \ref{lem:xi}.
We present a leader election algorithm  working in time $\tau$, if some $\lambda$-valent advice of sufficient size is available to the nodes.
The node $r$ of a tree $T_1 \in \cT$ is chosen as the leader.

Since $k_1=1$, and $k_2=2$, every node has to identify whether it belongs to the subtree $Q^1$ or $Q^2$, and has to learn its position in the respective subtree.
After doing this, it can output the path to leader, using the colored map.
At a high level, each node in $P^1$ is assigned the color $c_1$ and each node in $P^2$ is assigned the color $c_2$. Every node can identify the subtree to which it belongs by seeing the advice provided to the nodes in $P^1$ or $P^2$. The advice at the attached $\lfloor \log D \rfloor$ nodes, coding the distance from $r$, helps to find the positions of the nodes in the respective subtree, as explained before.
\end{proof}

Lemmas \ref{lem:large1} and \ref{lem:xi1} imply the following theorem.

\begin{theorem}
Let $D$ and $n'$ be positive integers such that $D =cn'+o(n')$, for some positive constant $c<1$. Let $\lambda>1$ be a constant integer. There exists a real $\beta_1$, $0 < \beta _1 <\frac{1}{2}$, depending only on $c$ and $\lambda$, such that for any algorithm  ELECT which solves election in $\tau=\lfloor \beta D \rfloor$ rounds with $\lambda$-valent advice, for every  tree with $\lambda$-election index at most $\tau$, for any constant $\beta<\beta _1$,  there exists
a tree with $\lambda$-election index at most $\tau$ with diameter $D$ and $n= n'+o(n')$ nodes, for which algorithm ELECT with $\lambda$-valent advice, working in time $\tau$, requires advice of size $\Omega(n)$.
\end{theorem}

\subsubsection{The algorithm}

For the second (positive) result,  we propose an election algorithm, working for any $n$-node tree of diameter $D$, with $\lambda$-valent advice of constant size. Our algorithm works in time $\tau=\lfloor \beta D \rfloor$, for trees with $\lambda$-election index at most $\tau$, for any constant $\beta>\beta_2$, where $\beta_2$ satisfies the equation
$\beta_2=2(\frac{1}{2}-\beta_2+2\epsilon) (\log_{\lambda} \left( \frac{1-\frac{c}{2}+\epsilon}{\frac{c}{2}-\beta_2c}\right)+1)$, with $\epsilon = \frac{1-c}{200}$.
First, we propose an algorithm working in time $\tau=\lfloor \beta D \rfloor$, that solves election for trees with $\lambda$-election index at most $\tau$, using $(\lambda+5)$-valent advice of constant size. Later, we show how to modify the algorithm, so that  $\lambda$-valent advice of constant size is enough. Let $\tau'=\lfloor \beta_2 D \rfloor$. Equivalently, we show an algorithm working in time $\tau=\gamma \tau'$, for any constant $\gamma>1$.

Let $T$ be a rooted $n$-node tree of diameter $D$ with $\lambda$-election index at most $\tau$. If $D$ is even, then the root $r$ is the central node, and if $D$ is odd, the root $r$ is one
of the endpoints of the central edge. This is the node that the algorithm will elect. The height of the tree is $\lceil \frac{D}{2} \rceil$.
At a high level, the advice is assigned to each of the nodes in $T$ in two steps. In the first step, certain nodes in $T$ of different depths are marked using five additional markers. We will later specify how these markers are coded. The coding will also include information about the direction from the marked node to the root. In the second step, all the non-marked nodes are assigned advice in such a way that the advice strings in these nodes collected in a specified way represent the sequence of port numbers from a marked node to the leader. In what follows, we use an integer parameter $k$, which will be defined later.
Let $L$ be initialized to the set of all leaves in $T$. The marking of the nodes in $L$ is done
 using five markers, $white$, $green$, $blue$, $red$ and $black$, as described below.
\begin{enumerate}[a.]
\item Mark the root $r$ with the marker {\it white}.
\item Let $v$ be the node in $L$ of largest depth. If $v$ has an ancestor $u$ at distance at most $\tau $ with $M(u)=white$ or $v$ has an ancestor $u$ at distance $<(k-2)\lfloor\frac{\tau}{k}\rfloor$ with $M(u)=green$ or $M(u)=blue$, then remove $v$ from $L$. Otherwise, let $u$ be the ancestor of $v$ at distance $(k-2)\lfloor\frac{\tau}{k}\rfloor$. Mark $u$ with the marker $blue$. If $v$ is not a leaf, mark it with the marker {\it green}. Add $u$ to $L$ and remove $v$ from $L$.
\item If $L$ is non-empty, go to Step b.
\item For every path of length $(k-2)\lfloor\frac{\tau}{k}\rfloor$ whose top node is green or blue, which ends with a green node or a leaf, and does not have any blue internal node, mark every $\lfloor\frac{\tau}{k}\rfloor$-th internal  node in this path, from top to bottom, with the marker {\it red}.

\item For every path of length $<(k-2)\lfloor\frac{\tau}{k}\rfloor$ whose top node is green or blue, which ends with a blue node or a leaf, and does not have any blue internal node, mark every $\lfloor\frac{\tau}{k}\rfloor$-th internal node $w$ in this path, from the top  to bottom, with the marker {\it black}, if this node is not already marked red and the distance between $w$ and $v$ is at least $2k+10$.
\end{enumerate}

\begin{figure}[h]
\centering
\includegraphics[width=0.5\textwidth]{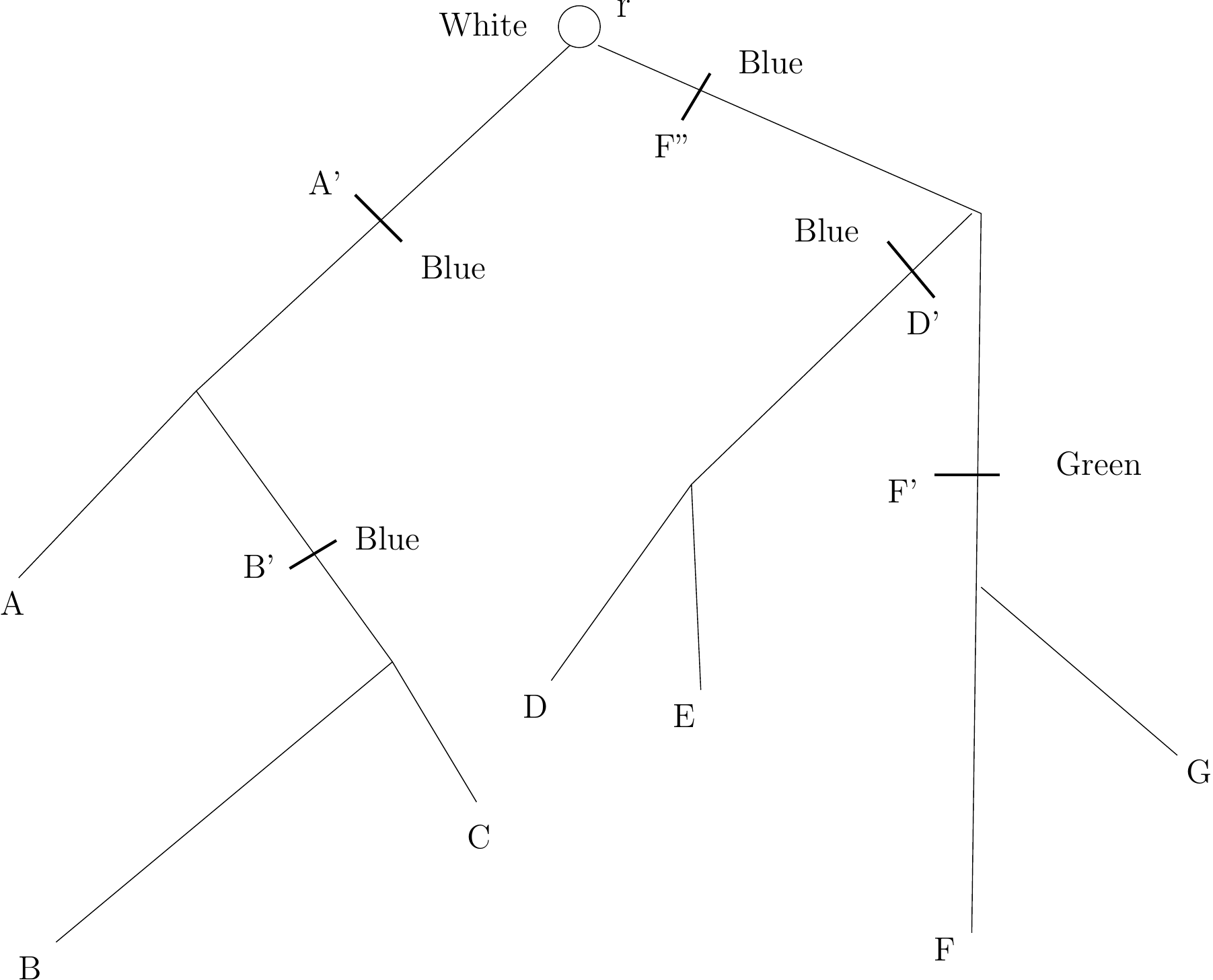}
\caption{Marking by blue and green markers}
\label{marking}
\end{figure}

Fig. \ref{marking} and Fig. \ref{coding} show an example of a tree $T$ and the marking of the nodes of $T$. More precisely, Fig. \ref{marking} shows the marking of the nodes of $T$ by the markers blue and green. Initially, the set $L$ contains the nodes $A$, $B$, $C$, $D$, $E$, $F$, and $G$. First, the root is marked white.
Then the nodes $A'$, $B'$, and $F'$ are marked blue and included in $L$. They are at distance $(k-2)\lfloor \frac{\tau}{k}\rfloor$ from the nodes $A$, $B$, and $F$, respectively. The other leaves of $T$ are removed from the set $L$ because one of the nodes $A'$, $B'$, and $F'$ is an ancestor of each of these leaves, at distance less than $(k-2)\lfloor \frac{\tau}{k}\rfloor$. The root $r$ is the ancestor of  $A'$ and $B'$ at distance less than $(k-2)\lfloor \frac{\tau}{k}\rfloor$, hence these two nodes are removed from $L$. The node $F''$ is marked blue. It is an ancestor of $F'$ at distance $(k-2)\lfloor \frac{\tau}{k}\rfloor$. Hence the mark of $F'$ is changed to green and $F'$ is removed from $L$. Finally, $F''$ is removed from $L$ and $L$ becomes empty. Fig. \ref{coding} shows the marking of red and black nodes on an example of two paths. Every $\lfloor \frac{\tau}{k}\rfloor$-th node of the path $A'$ to $A$ is marked red. Every $\lfloor \frac{\tau}{k}\rfloor$-th node of the path from $A'$ to $B$, which is not marked red, gets the mark black.

\begin{figure}[h]
\centering
\includegraphics[width=0.5\textwidth]{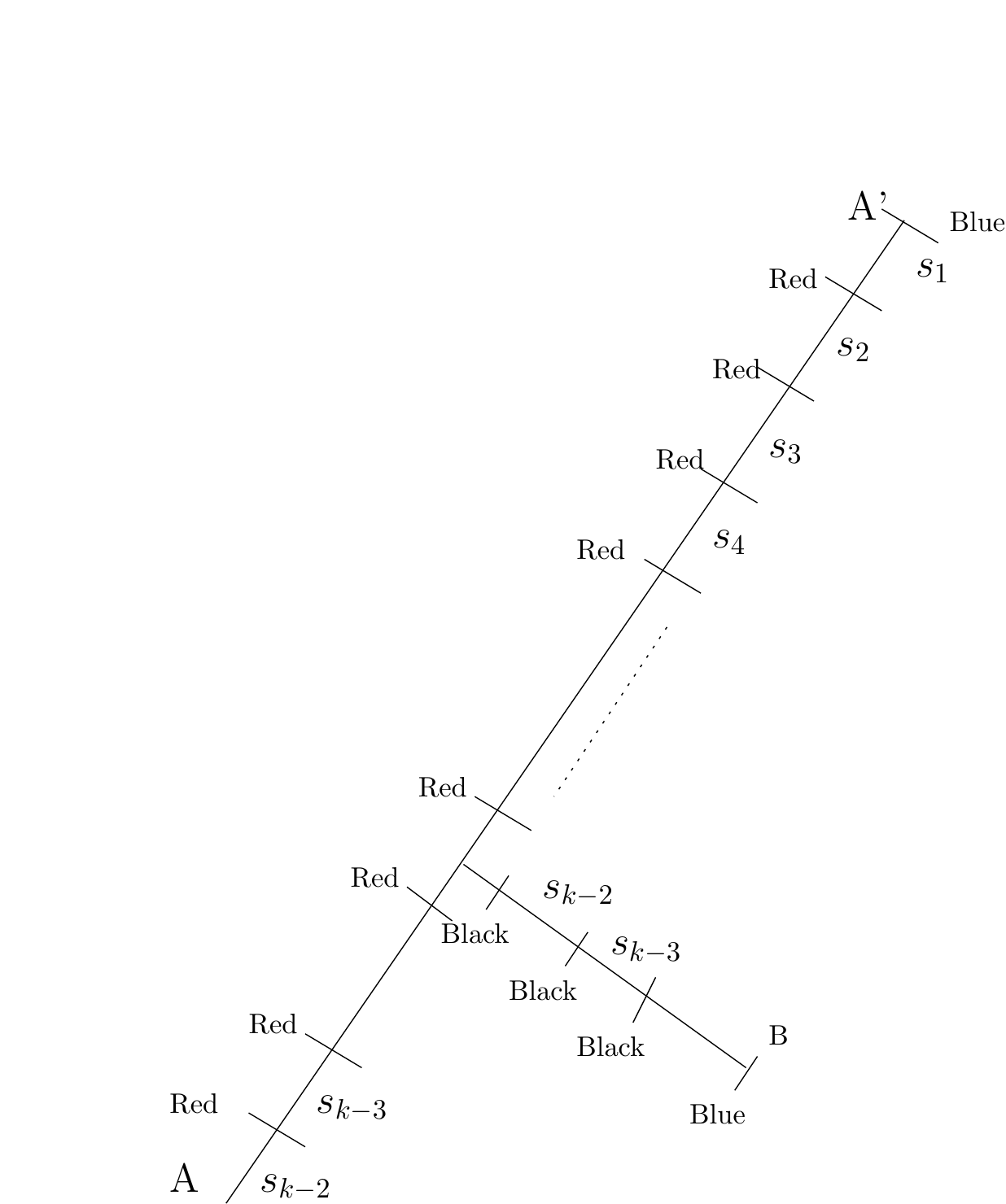}
\caption{Marking by red and black markers and the advice assignment at non-marked nodes}
\label{coding}
\end{figure}

We give a formal description of the above marking in Algorithm \ref{alg:marking}.
%The root $r$ is marked with the color white. Every node, which sees a white node in time $\tau$,
%choose it as the leader and computes the path to the leader by seeing its labeled ball.

According to Algorithm \ref{alg:marking}, the following statements are true.
\begin{enumerate}
\item  Every node $v$ in $T$ at depth at least $\tau+1$  has a blue or green ancestor at distance at most $(k-2)\lfloor\frac{\tau}{k}\rfloor$.
\item Every non-leaf node $u$, which is either green or blue, has at least one green descendant $u'$ at distance $(k-2)\lfloor\frac{\tau}{k}\rfloor$, and the path from $u$ to $u'$ does not contain any blue internal node.
\item The depth of every node $u$, which is either green or blue, is at most $\lceil\frac{D}{2}\rceil-(k-2)\lfloor\frac{\tau}{k}\rfloor$.
\end{enumerate}

 \begin{algorithm}
\caption{\textsc{Marking$(T,\tau)$}}
\begin{algorithmic}[1]
\STATE{$M(r) \leftarrow white$}
\STATE{$L\leftarrow$ \{$v|$, $v$ is a leaf of $T$\}.}
\WHILE{$L$ is non-empty}
\STATE{$v \leftarrow$ the node in $L$ with largest depth.}
\IF{$v$ has an ancestor $u$ at distance at most $\tau $ with $M(u)=white$ or $v$ has an ancestor $u$ at distance $<(k-2)\lfloor\frac{\tau}{k}\rfloor$ with $M(u)=green$ or $M(u)=blue$}
\STATE{$L=L\setminus\{v\}$.}
\ELSE
\STATE{$u \leftarrow$ the ancestor of $v$ at distance $(k-2)\lfloor\frac{\tau}{k}\rfloor$.}
\STATE{$M(u)\leftarrow blue$, $M(v)\leftarrow green$.}
\STATE{$L=L \cup \{u\}\setminus\{v\}$.}
\ENDIF
\ENDWHILE
\FOR{every path $P$ from $u$ to $v$ of length $(k-2)\lfloor\frac{\tau}{k}\rfloor$, where ($M(u)=blue$ or $M(u)= green$) and $M(v)=green$, and $v$ is a descendant of $u$}
\FOR{ every internal node $w \in P$ at depth $depth(u)+i\lfloor\frac{\tau}{k}\rfloor$, $1\le i\le k-3$}
\STATE{$M(w)\leftarrow red$.}
\ENDFOR
\ENDFOR
\FOR{every path $P$ from $u$ to $v$ of length $<(k-2)\lfloor\frac{\tau}{k}\rfloor$, where ($M(u)=blue$ or $M(u)= green$) and ($M(v)=green$ or $v$ is a leaf), and
there is no blue node between $u$ and $v$ in $P$, and $v$ is a descendant of $u$}
\FOR{every node $w \in P$ at depth $depth(u)+i\lfloor\frac{\tau}{k}\rfloor$, $i> 0$}
\IF{$M(w) \ne red$ and the distance between $w$ and $v$ is at least $2k+10$}
\STATE{$M(w)\leftarrow black$.}
\ENDIF
\ENDFOR
\ENDFOR
\end{algorithmic}\label{alg:marking}
\end{algorithm}

According to Algorithm \ref{alg:marking}, there can be two kinds of paths whose top node is green or blue, of length at most $(k-2)\lfloor\frac{\tau}{k}\rfloor$. The first kind of paths are of  length exactly $(k-2)\lfloor\frac{\tau}{k}\rfloor$ and end with a green node. All internal marked nodes on this path are red.
The second kind of paths are of length less than $(k-2)\lfloor\frac{\tau}{k}\rfloor$, end with a non-marked leaf or a blue node, and there is no internal blue node on the path. All marked internal nodes on this path are black or red.

The advice $C(u)$ of a node $u$ which is not previously marked is assigned depending on which kind of path it belongs to.
Below we explain the high-level idea of the advice assignment and its interpretation for each type of paths.

Consider a path $P$ of the first kind, of length $(k-2)\lfloor\frac{\tau}{k}\rfloor$. Let $P$  start with the top node $u$, $M(u)=green$ or $M(u)= blue$ and end with a node $v$, $M(v)=green$. All the nodes in $P$ are used to code the sequence of port numbers that represents the path from $u$ to the root $r$. Every node which belongs to such a path, can see the entire path in $\tau$ and thus can decode the path to the leader by collecting the advice from top to bottom.

Consider a path $P$ of the second kind, of length smaller than $(k-2)\lfloor\frac{\tau}{k}\rfloor$.
Let $P$ start with the top node $u$, $M(u)=green$ or $M(u)=blue$, end with a node $v$ which is either blue or a non-marked leaf, and contains $h$ black internal nodes and no blue internal nodes. Let $u'$ be the closest green descendant of $u$. In time $\tau$, any node $v$ in $P$ can see at least $k-h-1$ red nodes on the path from $u$ to $u'$. The advice given to non-marked nodes in the path from $u$ to $u'$ codes the path from $u$ to $r$. The node $v$ can see all nodes between these $k-h-1$ marked nodes. It concatenates top down the advice strings given to these nodes, obtaining a string $x$. The other part $y$ of the desired string is coded in the path $P$. The node $v$ obtains the string $y$ by collecting the advice from non-marked nodes of $P$. Finally, the node $v$ decodes the path from $v$ to the leader, using the concatenation $yx$.

 Fig. \ref{coding} shows the advice assignment at the non-marked nodes of the path from $A'$ to $A$, which is a path of the first kind. The advice assignment at the non-marked nodes between the black nodes of the path from $A'$ to $B$, which is a path of the second kind,  is also shown in this figure.

Below we give the pseudocode of the Algorithm \textsc{Coding$(T,\tau)$}, used to produce the advice for each node. The algorithm calls one of the two procedures specified in Algorithms \ref{alg:codingp1}, and \ref{alg:codingp2},  depending on the kind of path to which the node belongs.

\begin{algorithm}
\caption{\textsc{Coding$(T,\tau)$}}
\begin{algorithmic}[1]
\FOR{ every path $P$  of the first kind from node $u$ to $v$ of length $(k-2)(\lfloor \frac{\tau}{k}\rfloor)$, such that $M(u)=green$ or $M(u)=blue$,  $M(v)=green$ or $v$ is a leaf,  $v$ is a descendant of $u$, and there is no blue internal node in $P$}
\STATE{{\textsc{CodingPath\_1$(T,\tau)$}}.}
\ENDFOR
\FOR{every path $P$  of the second kind, of length less than $(k-2)\lfloor\frac{\tau}{k}\rfloor$, from a node $u$ to a node $v$, where $M(u) =green$ or $M(u) =blue$, , $M(v)=blue$ or $v$ is a leaf, $v$ is a descendant of $u$, and $P$ does not contain any blue internal node}
\STATE{{\textsc{CodingPath\_2$(T,\tau)$}}}
\ENDFOR
\STATE{For all the nodes $v$, which are not assigned any advice yet, assign advice $C(v)=c_1$.}
\end{algorithmic}\label{alg:coding}
\end{algorithm}

\begin{algorithm}
\caption{\textsc{CodingPath\_1$(T,\tau)$}}
\begin{algorithmic}[1]
\STATE{ Let $P=(u,u_0^1,u_0^2,\cdots,u_0^{\lfloor\frac{\tau}{k}\rfloor-1},u_1$, $u_1^1,u_1^2,\cdots,u_1^{\lfloor\frac{\tau}{k}\rfloor-1},u_2, \cdots,u_{k-3}, u_{k-3}^1,u_{k-3}^2,\cdots,u_{k-3}^{\lfloor\frac{\tau}{k}\rfloor-1},v)$,  such that $M(u)=green$ or $M(u)=blue$,  $M(v)=green$ or $v$ is a leaf,  $v$ is a descendant of $u$, and, for $1 \leq  j \leq k-3$, $M(u_j)=red$.}
\STATE{Let $s$ be the string of length $(k-2)(\lfloor \frac{\tau}{k}\rfloor-1)$ of $\lambda$ colors, that represents the sequence of port numbers from $u$ to $r$. Let $s_1$, $s_2$, $\cdots$, $s_{k-2}$ be the substrings of $s$ of same length such that $s$ is the concatenation $s_1s_2 \cdots s_{k-2}$.}
\FOR{$i=0$ to $k-3$}
\FOR{$j=1$ to $\lfloor \frac{\tau}{k}\rfloor-1$}
\STATE{$C(u_{i}^j) \leftarrow s_i(j)$}
\ENDFOR
\ENDFOR

\end{algorithmic}\label{alg:codingp1}
\end{algorithm}

\begin{algorithm}[h]
\caption{\textsc{CodingPath\_2$(T,\tau)$}}
\begin{algorithmic}[1]
\STATE{ Let $P'=(u_0,u_0^1,u_0^2,\cdots,u_0^{\lfloor\frac{\tau}{k}\rfloor-1},u_1$, $u_1^1,u_1^2,\cdots,u_1^{\lfloor\frac{\tau}{k}\rfloor-1},u_2, \cdots,u_{l-1}, u_{l-1}^1,u_{l-1}^2,\cdots,u_{l-1}^{\lfloor\frac{\tau}{k}\rfloor-1},u_l)$ be the subpath of $P$ where $l<k-2$ and, for $0\leq j \leq l$, $M(u_j)=black$.}
\FOR{$i=0$ to $l-1$}
\FOR{$j=1$ to $\lfloor \frac{\tau}{k}\rfloor-1$}
\STATE{$C(u_{i}^j) \leftarrow s_{k-i-2}(j)$}
\ENDFOR
\ENDFOR
\end{algorithmic}\label{alg:codingp2}
\end{algorithm}

\newpage

We now show how any node $v$ can compute the path to the leader $r$ from the advice seen in the labeled ball $\cB(v,\tau)$.
We consider the following cases, depending on what the node $v$ sees in the labeled ball $\cB(v,\tau)$.

\noindent
{\bf Case 1.} A node $v$ sees a white node.\\
Node $v$ outputs the sequence of port numbers from itself  to the white node.

\noindent
{\bf Case 2.} A node $v$ does not see a white node but sees a path of length $(k-2)\lfloor\frac{\tau}{k}\rfloor$ whose top node is green or blue, which ends with a green node or a leaf, and does not contain any blue internal node.
Let $P=(u,u_0^1,u_0^2,\cdots,u_0^{\lfloor\frac{\tau}{k}\rfloor-1},u_1$, $u_1^1,u_1^2,\cdots,u_1^{\lfloor\frac{\tau}{k}\rfloor-1},u_2 \cdots,u_{k-3}, u_{k-3}^1,u_{k-3}^2,\cdots,u_{k-3}^{\lfloor\frac{\tau}{k}\rfloor-1},u')$ be a path of the first kind, seen by node $v$, such that $M(u)=green$ or $M(u)=blue$, $M(u')=green$ or $u'$ is a leaf, $u$ is an ancestor of $u'$, and, for $1 \leq  j \leq k-2$, $M(u_j)=red$.  The node $v$ computes the sequence $s$ which is the concatenation
$$C(u_0^1) C(u_0^2)\cdots C(u_0^{\lfloor\frac{\tau}{k}\rfloor-1})C(u_1^1) C(u_1^2)\cdots C(u_1^{\lfloor\frac{\tau}{k}\rfloor-1}) C(u_2^1)\cdots C(u_2^{\lfloor\frac{\tau}{k}\rfloor-1}) \cdots C(u_{k-3}^1) \cdots C(u_{k-3}^{\lfloor\frac{\tau}{k}\rfloor-1}).$$
This string $s$ unambiguously codes the sequence of port numbers corresponding to the path from $u$ to $r$. Let $\pi(u,r)$ be the sequence of port numbers corresponding to the path from $u$ to $r$ represented by $s$. If $u$ is an ancestor of $v$, then $v$ computes $\pi(v,u)$ by seeing $\cB(v,\tau)$. Then it outputs the path to the leader as $\pi(v,u)$ followed by $\pi(u,r)$. Otherwise, if $u$ is a descendant of $v$, then let $l$ be the distance from $u$ to $v$.
Let $\pi(u,r)$ be the sequence of port numbers corresponding to the path from $u$ to $r$ represented by $s$.  The node $v$ computes $\pi(v,r)$ by deleting the first $l$ port numbers from $\pi(u,r)$ and outputs it.

\noindent
{\bf Case 3.} Case 1 and Case 2 are false.\\
According to Algorithm \ref{alg:marking},  node $v$ sees at least one blue or green ancestor in time $\tau$. Let $u$ be the closest green or blue ancestor of $v$. Since, $M(u)=green$ or $M(u)=blue$, there exists at least one descendant $u'$ of $u$ at distance $(k-2)\lfloor\frac{\tau}{k}\rfloor$ such that, $M(u')$ is green or $u'$ is a leaf, and there is no blue internal node in the path from $u$ to $u'$.
Let $u_1$, $u_2$, $\cdots, u_h$ be the black nodes between $u$ and $v$. Let $u_i^1$, $u_i^2$, $\cdots$, $u_i^{\lfloor\frac{\tau}{k}\rfloor-1}$ be the nodes between the nodes $u_i$ and $u_{i+1}$,  for $1 \leq i \leq h-1$.
            The node $v$ sees at least the $k-h-1$ highest red nodes, $a_1,a_2,\dots a_{k-h-1}$, of the path from $u$ to $u'$. Let $a_i^1$, $a_i^2$, $\cdots$, $a_i^{\lfloor\frac{\tau}{k}\rfloor-1}$ be the non-marked nodes between the  nodes $a_i$ and $a_{i+1}$,  for $1 \leq i \leq k-h-2$. Let $a_0^1$, $a_0^2$, $\cdots$, $a_0^{\lfloor\frac{\tau}{k}\rfloor-1}$ be the non-marked nodes between the  nodes $u$ and $a_1$.
            The node $v$ computes the string $s'$ which is the concatenation
             $(C(u_1^1) C(u_1^2)\cdots C(u_1^{\lfloor\frac{\tau}{k}\rfloor-1}) C(u_2^1)\cdots C(u_2^{\lfloor\frac{\tau}{k}\rfloor-1})$ $\cdots$ $C(u_{h-1}^1) \cdots C(u_{h-1}^{\lfloor\frac{\tau}{k}\rfloor-1}) )^R$. Node $v$ computes the string $s''$ which is the concatenation $C(a_0^1) C(a_0^2)\cdots C(a_0^{\lfloor\frac{\tau}{k}\rfloor-1})C(a_1^1) C(a_1^2)\cdots C(a_1^{\lfloor\frac{\tau}{k}\rfloor-1})$ $\cdots$ $C(a_{k-h-2}^1) \cdots C(u_{k-h-2}^{\lfloor\frac{\tau}{k}\rfloor-1})$. Then node $v$ computes $s=s'' s'$.
The string $s$ unambiguously codes the sequence $\pi(u,r)$ of port numbers, corresponding to the path from $u$ to $r$.  The node $v$ computes the sequence $\pi(v,u)$ of port numbers corresponding to the path from $v$ to $u$, by seeing $\cB(v,\tau)$. Finally, $v$ outputs the sequence  $\pi(v,u)$ followed by $ \pi(u,r)$.

Below we give the pseudocode of Algorithm \textsc{Decoding$(v,\tau)$}, executed by a node $v$, which outputs the path from $v$ to the leader. The algorithm uses one of the two procedures specified in Algorithms \ref{alg:dcodingp1}, \ref{alg:dcodingp2}, depending on the labeled ball
$\cB(v,\tau)$.

\begin{algorithm}
\caption{\textsc{Decoding$(v,\tau)$}}
\begin{algorithmic}[1]
\IF{there exists a white node in $\cB(v,\tau)$}
\STATE{output the sequence of port numbers from $v$  to the white node.}
\ELSE
\IF{there exists a path $P$ in $\cB(v,\tau)$ of length $(k-2)\lfloor \frac{\tau}{k}\rfloor$ whose top node is green or blue, which ends with a green node or a leaf, and does not have any internal blue node}
\STATE{\textsc{DecodingPath\_1$(v,\tau)$}.}
\ELSE
\STATE{\textsc{DecodingPath\_2$(v,\tau)$}.}
\ENDIF
\ENDIF
\end{algorithmic}\label{alg:dcoding}
\end{algorithm}

\begin{algorithm}
\caption{\textsc{DecodingPath\_1$(v,\tau)$}}
\begin{algorithmic}[1]
\STATE{Let $P=(u,u_0^1,u_0^2,\cdots,u_0^{\lfloor\frac{\tau}{k}\rfloor-1}u_1$, $u_1^1,u_1^2,\cdots,u_1^{\lfloor\frac{\tau}{k}\rfloor-1},u_2 \cdots,u_{k-3}, u_{k-3}^1,u_{k-3}^2,\cdots,u_{k-3}^{\lfloor\frac{\tau}{k}\rfloor-1},u')$ be a path of length $(k-2)\lfloor \frac{\tau}{k}\rfloor$, seen by node $v$, such that $M(u)=green$ or $M(u)=blue$,  $M(u')=green$ or $u'$ is a leaf, $u$ is an ancestor of $u'$, and, for $1 \leq  j \leq k-3$, $M(u_j)=red$.}
\STATE{$s \leftarrow C(u_0^1) C(u_0^2)\cdots C(u_0^{\lfloor\frac{\tau}{k}\rfloor-1}) C(u_1^1)\cdots C(u_1^{\lfloor\frac{\tau}{k}\rfloor-1})$ $\cdots$ $C(u_{k-3}^1) \cdots C(u_{k-3}^{\lfloor\frac{\tau}{k}\rfloor-1})$}
\STATE{Let $s=(p_1,p_2, \cdots, p_q)^*$.}
\IF{$u$ is an ancestor of $v$}
\STATE{$\pi(v,u) \leftarrow$ the sequence of port numbers corresponding to the path from $v$ to $u$.}
\STATE{Output $\pi(v,u)$ followed by $(p_1,p_2, \cdots, p_q)$.}
\ELSE
\STATE{ let $l$ be the distance from $w$ to $u$.}
\STATE{$\pi\leftarrow(p_{l+1},p_{l+2}, \cdots , p_q)$.}
\STATE{Output the sequence $\pi$.}
\ENDIF
\end{algorithmic}\label{alg:dcodingp1}
\end{algorithm}

\begin{algorithm}
\caption{\textsc{DecodingPath\_2$(v,\tau)$}}
\begin{algorithmic}[1]
\STATE{Let $u$ be the closest green or blue ancestor of $v$. Let $u_1$, $u_2$, $\cdots, u_h$ be the black nodes between $u$ and $v$. Let $u_i^1$, $u_i^2$, $\cdots$, $u_i^{\lfloor\frac{\tau}{k}\rfloor-1}$ be the nodes between the nodes $u_i$ and $u_{i+1}$,  for $1 \leq i \leq h-1$.}
\STATE{Let $a_1$, $a_2$, $\cdots, a_{k-h-1}$ be the $k-h-1$ highest red nodes in some path whose top node is $u$. Let $a_i^1$, $a_i^2$, $\cdots$, $a_i^{\lfloor\frac{\tau}{k}\rfloor-1}$ be the unmarked nodes between the nodes $a_i$ and $a_{i+1}$, for $1 \leq i \leq k-h-1$. Let $a_0^1$, $a_0^2$, $\cdots$, $a_0^{\lfloor\frac{\tau}{k}\rfloor-1}$ be the unmarked nodes between the nodes $u$ and $a_1$.}
\STATE{$s' \leftarrow (C(u_1^1) C(u_1^2)\cdots C(u_1^{\lfloor\frac{\tau}{k}\rfloor-1}) C(u_2^1)\cdots C(u_2^{\lfloor\frac{\tau}{k}\rfloor-1})$ $\cdots$ $C(u_{h-1}^1) \cdots C(u_{h-1}^{\lfloor\frac{\tau}{k}\rfloor-1}) )^R$.}
\STATE{$s'' \leftarrow C(a_0^1) C(a_0^2)\cdots C(a_0^{\lfloor\frac{\tau}{k}\rfloor-1}) C(a_1^1)\cdots C(a_1^{\lfloor\frac{\tau}{k}\rfloor-1})$ $\cdots$ $C(a_{k-h-2}^1) \cdots C(u_{k-h-2}^{\lfloor\frac{\tau}{k}\rfloor-1})$.}
\STATE{$s \leftarrow s's''$.}
\STATE{Let $s=(p_1,p_2, \cdots, p_q)^*$.}
\STATE{$\pi(v,u) \leftarrow$ the sequence of port numbers corresponding to the path from $v$ to $u$.}
\STATE{Output $\pi(v,u)$ followed by $(p_1,p_2, \cdots, p_q)$.}

\end{algorithmic}\label{alg:dcodingp2}
\end{algorithm}

\newpage

Note that the $(\lambda+5)$-valent advice described in Algorithms \ref{alg:marking} and \ref{alg:coding} has constant size: indeed, each of the five markers
can be coded in constant size, and each of the unmarked nodes gets one of the $\lambda$ colors as advice, which is also of constant size, since $\lambda$ is constant.
%Below we show how to solve election with $\lambda$-valent advice instead of $(\lambda+5)$-valent advice.

%\begin{lemma}
%The size of the advice provided to every node is constant.
%\end{lemma}
%\begin{proof}
%The advice provided to node $v$ ids either its marker, or it is the string $C(v)$, if the node is non-marked. In the first step, certain nodes are assigned some markers. Since there are constant number of markers are used for the marking, the size of of the advice to produce the markers are constant. To the all other nodes, the advice provided is one of the $\lambda$ colors which are of constant size. Hence, the advice provided to each of the nodes is constant.
%\end{proof}

%\vspace*{1cm}

%{\bf Algorithm with $\lambda$-valent advice}

We now describe how to code the five markers needed in Algorithm \ref{alg:dcoding}, using $\lambda$ colors. Instead of a single node, a sequence of constant length of consecutive nodes is used to code each marker in such a way that every node seeing the advice given to nodes of this sequence can identify a marker and can detect the direction to the root. We give the description for the case when $\lambda=2$. In this case, every node gets a single bit as advice.  For $\lambda>2$ the solution is similar.

Let $\tau=\gamma \tau'$, where $\gamma>1$ is any positive real constant. There exists an integer constant $k>3$, such that $\gamma \ge \frac{k+1}{k-3}$. We consider the time  $\tau \ge k(k-2)(2k+10)+(k+1)$ and $n\ge \frac{200}{1-c}$.

Let $s(v,r)^*$ be the binary string coding the sequence $s(v,r)$ of port numbers, corresponding to the path from $v$ to $r$. We compute the binary sequence $s'$ from $s(v,r)^*$ as follows. Let $s_1$, $s_2$, $\cdots$, $s_m$ be the substrings of $s(v,r)^*$ such that $|s_i|=k$, for $1 \le i <m$, $s_m \le k$, and  $s(v,r)^*$ is the concatenation $s_1s_2\dots s_m$. The binary string $s'$ is the concatenation $s_10s_20s_30\cdots s_{m-1}0s_m$. Note that, the binary string $1^{k+1}$ is never a substring of $s'$. We use the substring $1^{k+1}$ to code the markers as follows. Each marker is formed by a sequence of $2k+9$ consecutive nodes. We will prove that these sequences are disjoint for different markers.

\noindent
{\bf White marker}: All the nodes of a path starting from the root $r$ to a node of depth $2k+8$ are used to code the white marker. Consecutive $2k+9$ nodes, starting from the node $r$ to a node at depth $2k+8$, are assigned the advice in the following way. Give the node $r$ the advice 0. The next $k+1$ nodes are assigned the advice 1. The next five nodes are assigned the advice 1,0,1,0,0, respectively. The next $k+1$ nodes are assigned the advice 1. The last node at depth $2k+8$ is assigned the advice 0.

\noindent
{\bf Green marker}: Let $u$ be a node in $T$ that gets the marker green, if Algorithm \ref{alg:marking} is applied on $T$. For all the paths $P$ from $u$ to a node $v$, of length $(k-2)(\lfloor\frac{\tau}{k}\rfloor)$, such that $v$ is a descendant of $u$, $M(v)=green$, or $v$ is a leaf, and there is no blue internal node in the path, consecutive $2k+9$ nodes  of $P$ starting from $u$ as the top node are used to code the green marker. Give the first node $u$ the advice 0. The next $k+1$ nodes are assigned the advice 1. The next five nodes are assigned the advice 1,0,0,0,0, respectively. The next $k+1$ nodes are assigned the advice 1. The last node at depth $depth(u)+2k+8$ in $P$ is assigned the advice 0.

\noindent
{\bf Blue marker}: Let $u$ be a node in $T$ that gets the marker blue, if Algorithm \ref{alg:marking} is applied on $T$. For all the paths $P$ from $u$ to a node $v$, of length $(k-2)(\lfloor\frac{\tau}{k}\rfloor)$, such that $v$ is a descendant of $u$, $M(v)=green$, or $v$ is a leaf, and there is no blue internal node in the path, consecutive $2k+9$ nodes of $P$ starting from $u$ as the top node are used to code the blue marker. Give the first node $u$ the advice 0. The next $k+1$ nodes are assigned the advice 1. The next five nodes are assigned the advice 1,1,0,0,0, respectively. The next $k+1$ nodes are assigned the advice 1. The last node at depth $depth(u)+2k+8$ in $P$ is assigned the advice 0.

\noindent
{\bf Red marker}: Let $u$ be a node on a path $P$ of the first kind that gets the marker red, if Algorithm \ref{alg:marking} is applied on $T$.
Consecutive $2k+9$ nodes of $P$ starting from $u$ as the top node are used to code the red marker.
Give the first node $u$ the advice 0. The next $k+1$ nodes are assigned the advice 1. The next five nodes are assigned the advice 1,1,1,0,0, respectively. The next $k+1$ nodes are assigned the advice 1. The last node at depth $depth(u)+2k+8$ in $P$ is assigned the advice 0.

\noindent
{\bf Black marker}: Let $u$ be a node on a path $P$ of the second kind that gets the marker black, if Algorithm \ref{alg:marking} is applied on $T$.
Consecutive $2k+9$ nodes of $P$ starting from $u$ as the top node are used to code  the black marker.
Give the first node $u$ the advice 0. The next $k+1$ nodes are assigned the advice 1. The next five nodes are assigned the advice 1,1,1,1,0, respectively. The next $k+1$ nodes are assigned the advice 1. The last node at depth $depth(u)+2k+8$ in $P$ is assigned the advice 0.

\begin{proposition}\label{disjoint}
Sequences of nodes forming different markers are disjoint.
\end{proposition}
\begin{proof}
Let $u$ and $v$ be different nodes marked by Algorithm \ref{alg:marking}. Let $(u=u_0,u_1,u_2,\cdots,u_{2k+8})$ be a sequence of consecutive nodes with the top node $u$, and let
$(v=v_0,v_1,v_2,\cdots,v_{2k+8})$  be a sequence of consecutive nodes with the top node $v$. We prove that $u_i \ne v_j$, for $0\le i,j \le 2k+8$.

According to Algorithm \ref{alg:marking}, the distance between two red markers, two black markers, two green markers, a green marker and a red marker, the white and a green marker, the white and a blue marker, the white and a red marker, the white and a black marker, a red and a black marker, a black and a green marker, a red and a blue marker, is at least $\lfloor \frac{\tau}{k}\rfloor$. Since $\lfloor \frac{\tau}{k}\rfloor > 2k+9$, therefore $u_i \ne v_j$, for $0\le i,j \le 2k+8$ in these cases.

Let $u$ and $v$ be both marked blue by Algorithm \ref{alg:marking}. According to this algorithm, every blue marker has a descendant at distance $(k-2)\lfloor \frac{\tau}{k}\rfloor$, which is either green or a leaf, and there is no blue marker in the path to this descendant. Let $u'$ be such a descendant of $u$ and let $v'$ be such a descendant of $v$. Hence, $u$ (respectively $v$) does not belong to the path from $v$ to $v'$ (respectively from $u$ to $u'$). The sequence of $2k+9$ consecutive nodes forming the blue marker corresponding to the node $u$ (respectively $v$), with $u$ (respectively $v$) as the top node, is a part of the path from $u$ to $u'$ (respectively from $v$ to $v'$). Since the two paths from $u$ to $u'$ and from $v$ to $v'$ are disjoint, therefore, $u_i \ne v_j$, for $0\le i,j \le 2k+8$. If one of the nodes $u$ and $v$ is marked green and the other is marked blue by Algorithm \ref{alg:marking}, then the argument is similar as above.
%
%The nodes marked red by Algorithm \ref{alg:marking}  are on a path of the first kind, that does not contain any blue marker, and the nodes marked black by Algorithm \ref{alg:marking}  are on a path of the second kind. Since these two kinds of paths can intersect in at most one node, we have  $u_i \ne v_j$, for $0\le i,j \le 2k+8$, if $u$ is marked red and $v$ is marked black or vice-versa, if $u$ is marked red and $v$ is marked blue or vice-versa, and if $u$ is marked green and $v$ is marked black or vice-versa.

The remaining case is when one of the nodes $u$ and $v$ is marked black by Algorithm \ref{alg:marking}, and the other is marked blue.
According to Algorithm \ref{alg:marking}, the distance between a blue marker and a black marker is at least $2k+10$. Hence,  $u_i \ne v_j$, for $0\le i,j \le 2k+8$, if $u$ is marked black and $v$ is marked blue or vice-versa.
\end{proof}

In view of Proposition \ref{disjoint}, every node $v$ can unambiguously identify the markers by seeing the advice given to nodes in $\cB(v,\tau)$, as explained below.
 Consider a sequence of $2k+9$ consecutive nodes, with advice $01^{k+1}a_1\dots a_{k+7}$ at these consecutive nodes, respectively.
 According to the marking strategy, if this sequence forms a marker, then the string of advice bits at these nodes must be one of the following.

\begin{enumerate}

\item $01^{k+1} 101001^{k+1}0$ or $01^{k+1} 001011^{k+1}0$, where the first bit corresponds to a node $z$ and the last bit corresponds to a node $z'$. The node $v$ identifies the marker as white in this case. If the node sees the first string, it identifies $z'$ as the root $r$. Otherwise, it identifies $z$ as the root $r$.

\item $01^{k+1} 100001^{k+1}0$ or $01^{k+1} 000011^{k+1}0$, where the first bit corresponds to a node $z$ and the last bit corresponds to a node $z'$. The node $v$ identifies the marker as green in this case.

\item $01^{k+1} 110001^{k+1}0$ or $01^{k+1} 000111^{k+1}0$, where the first bit corresponds to a node $z$ and the last bit corresponds to a node $z'$. The node $v$ identifies the marker as blue in this case.

\item $01^{k+1} 111001^{k+1}0$ or $01^{k+1} 001111^{k+1}0$, where the first bit corresponds to a node $z$ and the last bit corresponds to a node $z'$. The node $v$ identifies the marker as red in this case.

    \item The node sees the string $01^{k+1} 111101^{k+1}0$ or $01^{k+1} 011111^{k+1}0$. The node $v$ identifies the marker as black in this case.
\end{enumerate}
In cases 2.--5., if the node $v$ sees the first string, it identifies the direction to the root as $z'$ to $z$. Otherwise, it identifies the direction to the root as $z$ to $z'$.

Nodes that are not in a segment coding a marker are called {\em non-marker nodes}.
%Note that, according to Algorithm \ref{alg:dcoding}, in the last two cases identifying the direction is not needed.
The non-marker nodes of a  path of the first kind with the top node $u$ are used to code the path from
 $u$ to  the root $r$. After the coding of the markers described above, the number of non-marker nodes in a path of the first kind is
$(k-2)(\lfloor \frac{\tau}{k}\rfloor -2k-9)$.
The following two lemmas show that $(k-2)(\lfloor \frac{\tau}{k}\rfloor -2k-9)$ nodes are indeed sufficient to code the path from $u$ to $r$.

\begin{lemma}\label{lem:large}
Let $T$ be an $n$-node tree of diameter $D=cn+o(n)$, where $n \ge \frac{200}{1-c}$. Let $s(v,r)$ be the sequence of port numbers corresponding to the path $P$ from $v$ to $r$. Then the length of the binary sequence  $(s(v,r))^*$ coding $s(v,r)$  with $\lambda$ colors is at most $\tau'+1$  for
 every node $v\in T$ of depth at most $ \lceil\frac{D}{2}\rceil-\tau'$.
\end{lemma}
\begin{proof}
As before, we first do the proof for a tree with diameter $D=cn$.\\

$D=cn$ and $\tau'=\lfloor \beta_2 D \rfloor$ implies $\tau'=\lfloor \beta_2 cn \rfloor$.

We have
$$\beta_2=2(\frac{1}{2}-\beta_2+2\epsilon) (\log_{\lambda} \left( \frac{1-\frac{c}{2}+\epsilon}{\frac{c}{2}-\beta_2c}\right)+1),$$
where $\epsilon=\frac{1-c}{200}$.

$$\beta_2cn=2(\frac{cn}{2}-\beta_2 cn+2\epsilon c n) (\log_{\lambda} \left( \frac{n-\frac{cn}{2}+\epsilon n}{\frac{cn}{2}-\beta_2cn}\right)+1).$$

Since,   $\epsilon = \frac{1-c}{200}$ and $n \ge \frac{200}{1-c}$, therefore, $\epsilon n > 1$.
Hence,

$$\beta_2cn \ge 2(\frac{cn}{2}-\beta_2 cn+2) \left(\log_{\lambda} \left( \frac{n-\frac{cn}{2}+1}{\frac{cn}{2}-\beta_2cn}\right)+1\right).$$

$$\beta_2cn \ge 2\left((\left\lceil\frac{cn}{2}\right\rceil -1)-(\lfloor\beta_2 cn\rfloor +1)+2\right) \left(\log_{\lambda} \left( \frac{n-\frac{cn}{2}+1}{\frac{cn}{2}-\beta_2cn}\right)+1\right).$$

$$(\lfloor\beta_2cn\rfloor +1) \ge 2(\lceil\frac{cn}{2}\rceil-\lfloor\beta_2 cn\rfloor) \left(\log_{\lambda} \left( \frac{n-\lfloor\frac{cn}{2}\rfloor-\lfloor \beta_2 c n\rfloor}{\lceil\frac{cn}{2}\rceil-\lfloor\beta_2cn\rfloor}\right)+1\right).$$

$$(\tau'+1)\ge 2(\left\lceil \frac{D}{2}\right\rceil -\tau') \left(\log_{\lambda} \left( \frac{n-\lfloor\frac{D}{2}\rfloor-\tau'}{\lceil \frac{D}{2}\rceil-\tau'}\right )+1\right).$$

Let $v$ be a node of $T$ of depth at most $ \lceil\frac{D}{2}\rceil-\tau'$.
Let $P=(v,u_1,u_2,\cdots,u_{\ell-1},r)$ be the path from $v$ to $r$. Let $d(u)$ denote the degree of node $u$.

Since the depth of $v$ in $T$ is at most $\lceil\frac{D}{2}\rceil-\tau'$, there exist at least $\tau'$ nodes in $T$ with depth larger than $\lceil\frac{D}{2}\rceil-\tau'$.
Also, since the diameter of $T$ is $D$, there exists at least one path of length $\lfloor\frac{D}{2}\rfloor$ with no common node with $P$ other than $r$.
%Let $P=(v,u_1,u_2,\cdots,u_{k-1},r)$. Let $d(u)$ denote the degree of node $u$.
Then
$d(v)+ \sum_{i=1}^{\ell-1} d(u_i) \le n-\tau'-\lfloor\frac{D}{2}\rfloor$.
 The sum of logarithms of these degrees is $ \log_{\lambda} d(v) + \sum_{i=1}^{\ell-1} \log_{\lambda} d(u_i) = \log_{\lambda} \left(d(v)\prod_{i=1}^{\ell} d(u_i)\right)$. The value of $d(v)\prod_{i=1}^{\ell} d(u_i)$ is maximized when $d(v)=d(u_i)= \frac{d(v)+ \sum_{i=1}^{\ell} d(u_i)}{\ell}$ for $1\le i \le \ell-1$. Hence, this sum of logarithms is at most   $\ell \log_{\lambda} (\frac{n-\tau'-\lfloor\frac{D}{2}\rfloor}{\ell}) \le (\lceil\frac{D}{2}\rceil-\tau') \log (\frac{n-\tau'-\lfloor\frac{D}{2}\rfloor}{\lceil\frac{D}{2}\rceil-\tau'})$.

 Let $s(v,r)=(p_1,p_2,\cdots,p_\ell)$ be the sequence of port numbers corresponding to the path $P$ from $v$ to $r$. Then the length of $s(v,r)^*$ is at most $2\sum_{i=1}^\ell(\lfloor \log_{\lambda} p_i\rfloor +1) \le 2\sum_{i=1}^\ell( \log_{\lambda} p_i +1) \le 2(\lceil\frac{D}{2}\rceil-\tau') (\log_{\lambda} (\frac{n-\tau'-\lfloor\frac{D}{2}\rfloor}{\lceil\frac{D}{2}\rceil-\tau'})+1) \le \tau'+1$.

 Therefore, the length of the sequence  $(s(v,r))^*$ is at most $\tau'+1$,  for
 every node $v\in T$ of depth at most $ \lceil\frac{D}{2}\rceil-\tau$.

The generalization of the reasoning to the case $D=cn+o(n)$ follows from continuity arguments. It can be observed that the real $\beta_2$ in this case can be found
arbitrarily close to that derived for the case $D=cn$, for sufficiently large $n$.
\end{proof}

\begin{lemma}
The length of the binary sequence $s'$, computed from $s(v,r)^*$, where $v$ is
the highest node of a blue or a green marker, is at most $(k-2)(\lfloor \frac{\tau}{k}\rfloor -2k-9)$.
\end{lemma}
\begin{proof}
According to the marking strategy, the length of the path from $v$ to the root $r$ is at most
$\lceil \frac{D}{2}\rceil-(k-2)\lfloor\frac{\tau}{k}\rfloor \le \lceil \frac{D}{2}\rceil-(k-2)\frac{\tau}{k} \le \lceil \frac{D}{2}\rceil - \frac{(k-2)(k+1)}{k(k-3)}\tau' \le \lceil \frac{D}{2}\rceil-\tau'$.

%First, we show that The length of the binary sequence $s'$ from $s(v',r)^*$, where $v'$ is ia a node of depth $\lceil \frac{D}{2}\rceil-\tau'$ is at most $(k-1)(\lfloor \frac{\tau}{k}\rfloor -2k-9)$.
By Lemma \ref{lem:large}, the length of the binary string $s(v,r)^*$ is at most $\tau'+1$. The binary string $s'$ is computed from $s(v,r)^*$ by inserting a 0 after every $k$-th bit of $s(v,r)^*$. Hence, $|s'|\le \frac{k+1}{k}(\tau'+1)$. Now,

$(k-2)(\lfloor \frac{\tau}{k}\rfloor -2k-9) \ge (k-2)(\frac{\tau}{k} -2k-10) = \frac{k-2}{k}\tau-(k-2)(2k+10)= \frac{k+1}{k}(\frac{k-2}{k+1}\tau-\frac{k(k-2)(2k+10)}{k+1})$.

Since $\tau \ge k(k-2)(2k+10)+(k+1)$, we have $\frac{k(k-2)(2k+10)}{k+1}\le \frac{\tau}{k+1}-1$. Therefore,

$(k-2)(\lfloor \frac{\tau}{k}\rfloor -2k-9)\ge \frac{k+1}{k}(\frac{k-2}{k+1}\tau-\frac{\tau}{k+1}+1)\ge \frac{k+1}{k}(\frac{1}{\gamma}\tau+1)=\frac{k+1}{k}(\tau'+1)\ge |s'|$.
\end{proof}

We are now ready to formulate the two final algorithms working for $2$-valent advice: Algorithm \textsc{Bounded Valency Advice}$(T,\tau)$ which assigns $2$-valent advice
of constant size to all nodes, and Algorithm  \textsc{Bounded Valency Election$(v,\tau)$}, which uses this advice to perform election in time $\tau=\gamma \tau'$, where $\tau'=\lfloor\beta_2D\rfloor$ and $\gamma$ is any constant greater than 1. Let $k>3$ be an  integer such that $\gamma>\frac{k+1}{k-3}$.

\vspace*{1cm}
\noindent
{\bf \textsc{Bounded Valency Advice}$(T,\tau)$}\\
If $\tau < k(k-2)(2k+10)+k+1$ or $n<\frac{200}{1-c}$, then the advice is assigned to each node of the tree in the following way. Consider a mapping $f: V(T) \rightarrow \{0,1\}$. Let $T(f)$ be the  node-colored map of $T$ corresponding to the mapping $f$. Let $\cF=\{f| f: V(T) \rightarrow \{0,1\}\}$. Let $f'\in \cF$ be the mapping such that for any two nodes $v_1, v_2 \in V(T(f'))$, $\cB(v_1,\tau) \ne \cB(v_1,\tau)$. Such a mapping exists because $\xi_2(T) \leq \tau$, by assumption. The mapping $f'$ can be found by computing $\cB(v,\tau)$ for each node $v$, for all the mappings in $\cF$. (Recall that this work is done by the oracle, so exhaustive search can be used, because time does not matter here.)  The advice assigned to each node $v$ is the binary string unambiguously coding the couple $(T(f'),f'(v))$.

For $\tau \geq k(k-2)(2k+10)+k+1$ and $n\ge \frac{200}{1-c}$, the advice is assigned to each node of the tree in two steps. In the first step, one-bit advice is assigned to some nodes of the tree to code markers of five types, as explained at the beginning of this section.
In the second step, we apply  Algorithms \ref{alg:codingp1}, and \ref{alg:codingp2} in the following way.
 Consider segments of $2k+9$ consecutive nodes,  coding the markers of various colors. Call such segments green, red etc. , if the corresponding marker is green, red etc. Contract any such segment into one node, giving it the corresponding color, and apply Algorithms \ref{alg:codingp1} and \ref{alg:codingp2},
 with the following modification.
 The length of each subsequence $s$ is $(k-2)(\lfloor \frac{\tau}{k}\rfloor -2k-9)$ instead of
$(k-2)(\lfloor \frac{\tau}{k}\rfloor -1)$, and the length of the subsequences $s_i$ of $s$  is $(\lfloor \frac{\tau}{k}\rfloor -2k-9)$ instead of $(\lfloor \frac{\tau}{k}\rfloor -1)$.

\vspace*{1cm}
\noindent
{\bf \textsc{Bounded Valency Election$(v,\tau)$}}\\
If $\tau < k(k-2)(2k+10)+k+1$ or $n<\frac{200}{1-c}$, then the advice provided to the node $v$ is a binary code of $(T(f),f(v))$, for some mapping $f$. The node $v$ learns $\cB(v,\tau)$ in time $\tau$, and identifies its unique position in $T(f)$. Then it outputs the sequence of port numbers corresponding to the path from $v$ to $r$ by seeing this path in $T(f)$.

Otherwise, the advice provided to each node is either 0 or 1. We apply Algorithm \ref{alg:dcoding}, using which every node
performs leader election, in the following way. Consider segments of $2k+9$ consecutive nodes,  coding the markers of various colors.
The node $v$ unambiguously identifies these segments in $\cB(v,\tau)$, as explained above. Call such segments green, red etc., if the corresponding marker is green, red etc. Contract any such segment into one node, giving it the corresponding color, and apply Algorithm \ref{alg:dcoding} to this contracted tree, with the following
modification.
The length of each subsequence $s$ is $(k-2)(\lfloor \frac{\tau}{k}\rfloor -2k-9)$ instead of
$(k-2)(\lfloor \frac{\tau}{k}\rfloor -1)$, and the length of the subsequences $s_i$ of $s$  is $(\lfloor \frac{\tau}{k}\rfloor -2k-9)$ instead of $(\lfloor \frac{\tau}{k}\rfloor -1)$.

The following result estimates the size of advice given to the nodes by Algorithm \textsc{Bounded Valency Advice}$(T,\tau)$, whenever the allocated time $\tau$
is at least $\gamma\tau'$, for a constant $\gamma>1$.

\begin{theorem}
For any real constant $\gamma>1$, the advice assigned to each node $v$ of the tree $T$ by Algorithm \textsc{Bounded Valency Advice}$(T,\tau)$ is  of constant size.
\end{theorem}
\begin{proof}
First consider the case when $\tau < k(k-2)(2k+10)+k+1$ or $n<\frac{200}{1-c}$. Since $\gamma$ is constant, $k$ is also constant and hence $\tau$ is constant in this case. Since, $\tau=\gamma \tau'$, $\tau'=\lfloor \beta_2D\rfloor$, and $D=cn+o(n)$, therefore $D$ and $n$ are constant in this case. Hence, the size of the tree is constant. The advice provided to each node codes the colored map $T(f')$ and the bit  given to the node in the map, and hence this code is of constant size.

If $\tau \ge k(k-2)(2k+10)+k+1$ and $n \ge \frac{200}{1-c}$,  then according to Algorithm \textsc{Bounded Valency Advice}$(T,\tau)$, every node gets either 0 or 1 as advice.  Hence the advice provided to every node is of size 1 in this case.
%
%Consider the case when $\tau < k(k-1)(2k+10)+k+1$. Since, $\tau=\gamma \tau'$, $\tau'=\lfloor \beta_2D\rfloor$, and $D=cn+o(n)$, therefore, $\tau$, $D$ and $n$ are constant in this case. Hence, the size of the tree is constant. The advice provided to each node is the whole colored map and the color of the node in the map. With the whole colored map as advice, every node elect a common node as the leader if $\xi_\lambda(T) \le \tau$.
\end{proof}

We finally prove the correctness of Algorithm  \textsc{Bounded Valency Election$(v,\tau)$}.
%
%\begin{lemma}
%Let $\tau > k(k-1)(2k+10)+k+1$. According to the marking strategy described above, in time $\tau$, every node have at least one of the following views in $\cB(v,\tau)$.
%\begin{enumerate}
%\item Node $v$ sees a white marker.
%\item Node $v$ sees a path of length $(k-1)\lfloor\frac{\tau}{k}\rfloor+2k+9$ whose top $2k+9$ nodes code a green or blue marker, the lowest $2k+9$ nodes code a green marker and does not contain any blue marker in the path.
%\item Node $v$ sees a path of length $(k-1)\lfloor\frac{\tau}{k}\rfloor$ whose top $2k+9$ nodes code a green or blue marker, the lowest node is a leaf and does not contain any blue marker in the path.
%\item Node $v$ sees
%\end{enumerate}
%\end{lemma}

\begin{theorem}\label{th:correctness}
Every node $v$ of a tree $T$ with $\lambda$-election index at most $\tau$, executing  Algorithm  \textsc{Bounded Valency Election$(v,\tau)$} in time $\tau$,
outputs the sequence of port numbers corresponding to the path from $v$ to~$r$.
\end{theorem}
\begin{proof}
The proof is trivial when $\tau < k(k-2)(2k+10)+k+1$ or $n<\frac{200}{1-c}$. Suppose that
$\tau \ge k(k-2)(2k+10)+k+1$ and $n \ge \frac{200}{1-c}$.

\noindent
{\bf Case 1.} Node $v$ sees a white marker in  $\cB(v,\tau)$.

It identifies the root $r$ as one of the two extremities of this marker,  as explained above, by seeing the advice string on the nodes of the marker. The node $v$ outputs the sequence of port numbers corresponding to the  path from $v$ to $r$, reading it from $\cB(v,\tau)$.

\noindent
{\bf Case 2.} Node $v$ sees a path of length $(k-2)\lfloor\frac{\tau}{k}\rfloor+2k+9$ whose top $2k+9$ nodes form a green or a blue marker, the lowest $2k+9$ nodes form a green marker and which does not contain any internal blue marker.

Let $u$ be the highest node of the top blue or green marker and let $u'$ be the highest node of the bottom green marker.
According to Algorithm \ref{alg:codingp1}, the advice pieces at the non-marker nodes on the path from $u$ to $u'$, read top-down, form the binary sequence $s$ representing the sequence of port numbers corresponding to the path from $u$ to $r$. Since the node can see the entire path between $u$ and $u'$, it correctly computes $s$ by collecting the advice from the non-marker nodes top to bottom between these two markers. The node computes the correct path to the leader according to Algorithm \ref{alg:dcodingp1}.

\noindent
{\bf Case 3.} Node $v$ sees a path of length $(k-2)\lfloor\frac{\tau}{k}\rfloor$ whose top $2k+9$ nodes form a green or a blue marker, the lowest node is a leaf and which does not contain any blue marker.

Let $u$ be the highest node of the top blue or green marker and let $u'$ be the leaf. The path from $u$ to $u'$ is of the first kind, and hence the advice pieces at the non-marker nodes on the path from $u$ to $u'$, read top-down, form the binary sequence $s$ representing the sequence of port numbers corresponding to the path from $u$ to $r$. Therefore, by a similar argument as in Case 2, the node $v$ computes the correct path to the leader from itself.

\noindent
{\bf Case 4.} Case 1, Case 2, and Case 3 are not true.

According to the marking strategy, in time $\tau$ the node $v$ sees at least one green or blue marker, whose top node $u$ is an ancestor of $v$. Since $u$ is the top node of a green or a blue marker, according to Algorithm \ref{alg:marking}, there exists a descendant $u'$ at distance $(k-2)\lfloor\frac{\tau}{k}\rfloor$ from $u$,  such that $u'$ is a top node of a green marker or $u'$ is a leaf.
Let there be $h$ black markers in the path from $v$ to $u$.
In time $\tau$, $v$ can see at least $k-h-1$ highest red markers on the path from $u$ to $u'$.
Then, according to Algorithm \ref{alg:codingp2}, the advice pieces at the non-marker nodes between the $h$ black markers, read top-down, form the
concatenation $s_{k-h}s_{k-h+1}\cdots s_{k-2}$. Denote this concatenation by $s'$. The advice pieces at the
non-marker nodes between $u$ and these $k-h-1$ highest red markers, read top-down, form
the concatenation $s_1s_2\cdots s_{k-h-1}$. Denote this concatenation by $s''$. Therefore, the node $v$ computes the sequence $s$ correctly by computing $s''s'$. The node computes the correct path to the leader according to Algorithm \ref{alg:dcodingp2}.
\end{proof}

{\bf Remark.}
In the case of $\lambda$-valent advice for $\lambda>2$, no changes are needed in the solution given above, because markers can be coded using the first two colors among $c_1,c_2,\dots ,c_{\lambda}$, and the non-marker nodes get the advice prescribed by Algorithm  \ref{alg:coding}, which was formulated for any $\lambda \geq 2$.

\subsubsection{Estimating $\beta_2-\beta_1$}

In the two previous sections, we proved that for an $n$-node tree with diameter $D=cn+o(n)$, there exist two positive reals $\beta_1$ and $\beta_2$, depending only on $c$ and $\lambda$, such that any leader election algorithm, working in time $\tau=\lfloor \beta D \rfloor$, for any constant $\beta<\beta_1$ requires advice of size $\Theta(n)$ and there exists a leader election algorithm, working in time $\tau=\lfloor \beta D \rfloor$, for any constant $\beta>\beta_2$, with advice of constant size.
Hence, the time values $\beta_1D  \leq \tau \leq \beta_2D$ are the part of the time spectrum that our (tight) results do not cover. In this section we show that this unchartered  region is small. In particular, we show that  $\beta_2-\beta_1<1/8$, for any $c$ and $\lambda$.

\begin{figure}[h]
\centering
\includegraphics[width=0.6\textwidth]{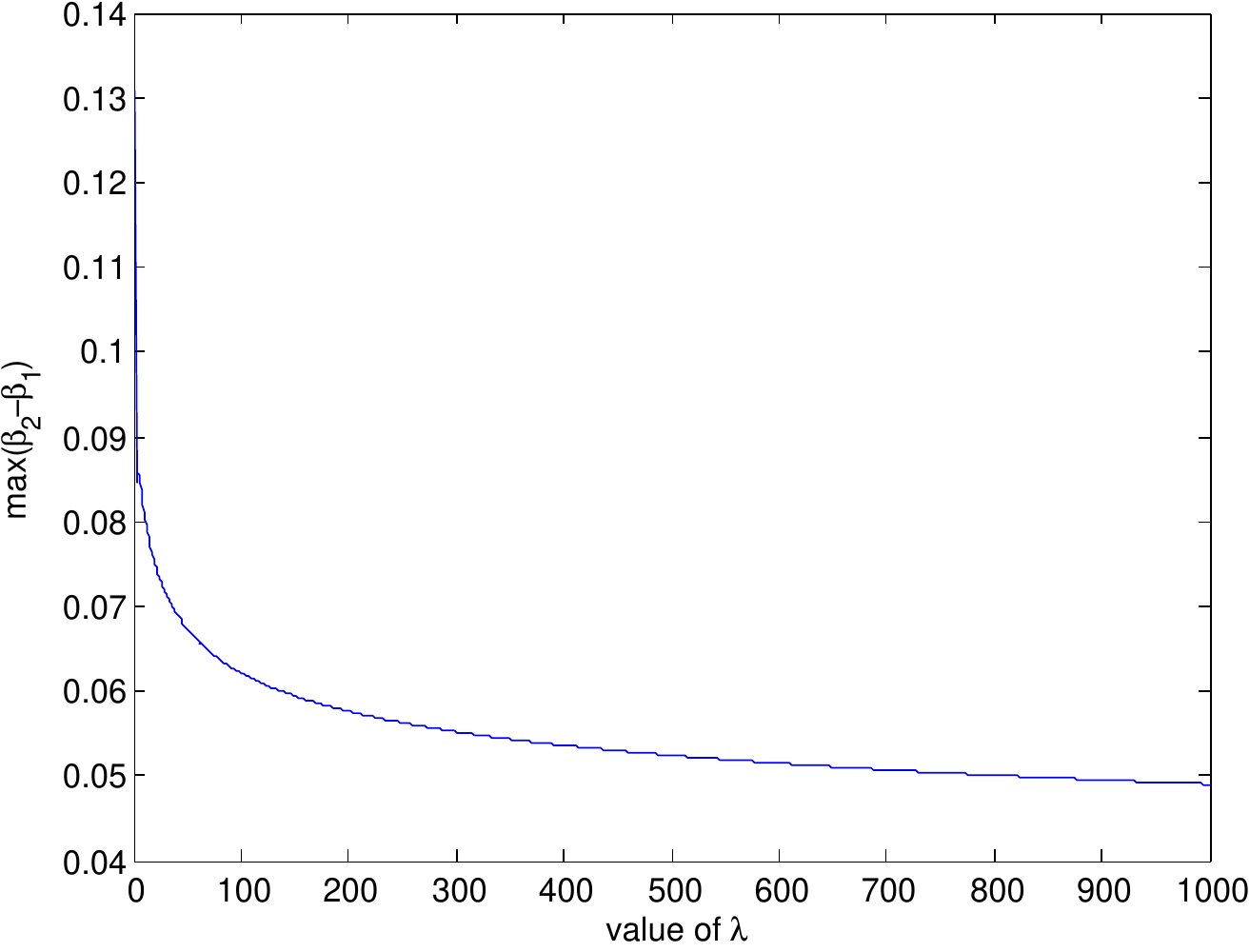}
\caption{Maximum value of $\beta_2-\beta_1$ for different values of $\lambda$}
\label{fig:lambda}
\end{figure}

\begin{figure}[h]
\centering
\includegraphics[width=0.6\textwidth]{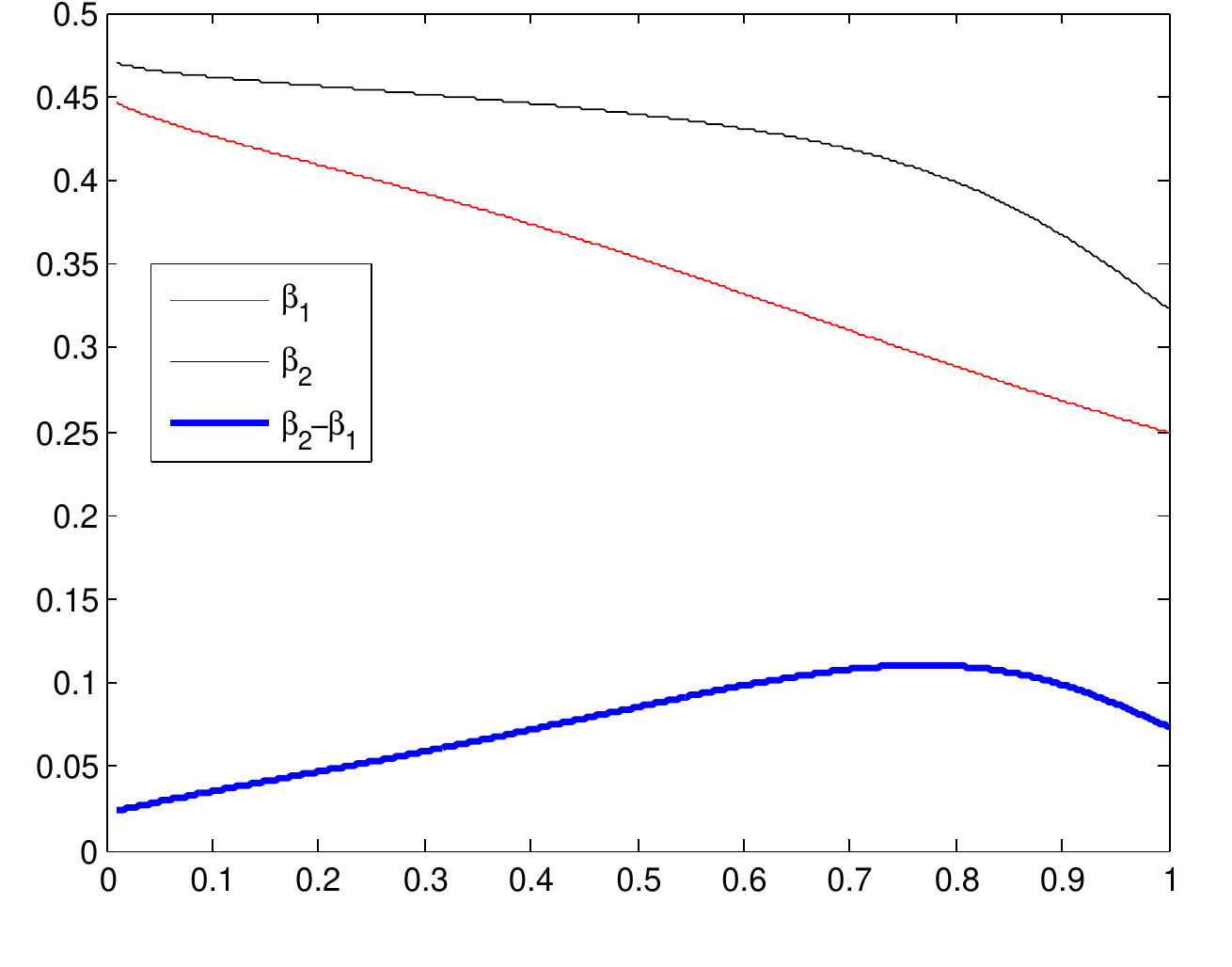}
\caption{Curves for $\beta_1$, $\beta_2$ and $\beta_2-\beta_1$ as functions of $c$, for $\lambda=2$}
\label{fig:lambda_2}
\end{figure}

We plot the two equations $\beta_1= \frac{1-2\beta_1}{2} \log_{\lambda} \left(\frac{\frac{1}{2}-\beta_1 c}{\frac{c}{2}-\beta_1 c+\epsilon}\right)$ and $\beta_2=2(\frac{1}{2}-\beta_2+2\epsilon) (\log_{\lambda} \left( \frac{1-\frac{c}{2}+\epsilon}{\frac{c}{2}-\beta_2c}\right)+1)$ using Matlab. We plot the maximum difference between $\beta_2$ and $\beta_1$, for different values of $c$ and for a fixed integer value of $\lambda$. Figure \ref{fig:lambda} shows that the maximum difference between $\beta_2$ and $\beta_1$, over all values of $c$, is decreasing as a function of $\lambda$. Hence we concentrate on the worst case $\lambda=2$.
Figure \ref{fig:lambda_2} shows the values of $\beta_2$, $\beta_1$ and $\beta_2-\beta_1$, for $\lambda=2$ as functions of $c$. The maximum value of $\beta_2-\beta_1$ is $0.1208<1/8$, and it is taken for $c$ approximately 0.8.

\section{Conclusion}

While our results for advice of unbounded valency are complete, for advice of constant valency they leave a small (sub-polynomial) gap for small diameter $D$.
Also, there are parts of the allocated time spectrum where our results do not work (in particular, for large diameter) although, as we argued, this unchartered territory is small.
Filling these gaps is a natural open problem. Also generalizing our results to arbitrary graphs remains open. As suggested by the comparison of results
in \cite{GMP} and \cite{DiPe}, where the minimum size of 1-valent advice sufficient for leader election was studied for trees and for graphs respectively, results
for arbitrary graphs may be very different from those for trees.

Another open problem is related to the precise definition of leader election in anonymous networks.
The definition adopted in this paper is the same as in \cite{GMP}, and requires each node to output a path to the leader.
It should be noted that, apart from this formulation of leader election in anonymous networks, there are two other possibilities involving weaker requirements.
The weakest of all is the requirement that the leader has to learn that it is a leader, and non-leaders have to learn that they are not, without  the necessity of learning who is the leader. The latter variant was called {\em selection} in \cite{MP}, and differences between election and selection were discussed there in a related context of finding the largest-labeled node
in a labeled graph. Of course, in our context, selection can be achieved in time 0 with 2-valent advice of size 1.

Even if all nodes have to learn who is the leader in an anonymous network,
one might argue that it is enough for every node to learn a port corresponding to a shortest path towards the leader, as then, e.g., for the task when all nodes have to send some data to the leader,
packets could be routed to the leader from node to node, using only this local information. This is indeed true, if nodes want to cooperate with others by revealing the
local port towards the leader when retransmitting packets. In some applications, however,  such a cooperation may be uncertain, and even when it occurs, it may slow down transmission, as the local port has to be retrieved from the memory of the relaying node.

It should be noted that, for trees and for $\lambda$-valent advice with $\lambda>2$, this weaker variation of leader election can be achieved in time 1 and with advice of size 2. It is enough to choose any node $r$ of the tree as the leader,
and give to every node its distance modulo 3 from the leader as advice. In time 1, every node at distance $x$ (mod 3) from $r$ can see a neighbor at distance $x-1$ (mod 3) from $r$ (its parent in the tree rooted at $r$), and output the port towards it. For $\lambda=2$, a similar solution can be obtained in time 2 and with advice of size 1.

%Hence there are (at least) three formulations of the leader election problem in anonymous networks, each suitable for different applications.
%In this paper we chose the strongest of all three, similarly as in \cite{GMP}. It may be interesting to investigate how our results would change for the two alternative formulations of leader election.

%%%%%%%%%%%%%%%%%%%%%%%%%%%%%%%%%%%%%%%%%%%%%%%%%%%%%%%%%%%
\bibliographystyle{plain}

%%%%%%%%%%%%%%%%%%%%%%%%%%%%%%%%%%%%%%%%%%%%%%%%%%%%%%%%%%%

\end{document}